\documentclass{amsart}

\usepackage[utf8]{inputenc}
\usepackage[T1]{fontenc}
\usepackage{amsfonts}
\usepackage{stmaryrd}
\usepackage{pxfonts}
\usepackage{microtype}
\usepackage{mathtools}
\usepackage{booktabs}
\usepackage{amssymb, latexsym}
\usepackage{amsmath}
\usepackage{eso-pic}
\usepackage{stackengine}
\usepackage[inline]{enumitem}
\usepackage{xspace}
\usepackage{tikz}

\usepackage[all]{onlyamsmath}

\makeatletter
	\let\@@tikzpicture\tikzpicture
	\def\tikzpicture{\catcode`\$=3 \@@tikzpicture}
\makeatother

\usetikzlibrary{arrows,automata,positioning}

\tikzset{shorten >=1pt, >=stealth, auto}
\tikzset{
    every state/.style={%
           rectangle,%
           rounded corners,%
           draw=black,%
           very thick,%
           minimum height=2em,%
           inner sep=2pt,%
           text centered,%
           },
}

\newtheorem{theorem}{Theorem}[section]
\newtheorem{lemma}[theorem]{Lemma}

\newtheorem{problem}{Problem}

\newtheorem{proposition}[theorem]{Proposition}

\newtheorem{corollary}[theorem]{Corollary}

\theoremstyle{definition}
\newtheorem{definition}[theorem]{Definition}

\theoremstyle{remark}
\newtheorem{remark}[theorem]{Remark}
\numberwithin{equation}{section}
\numberwithin{problem}{section}

\newcommand{\B}{{\mathbb B}}

\newcommand{\N}{{\mathbb{N}}}

\newcommand{\cl}{\mathrm{cl}}

\newcommand{\suf}[1]{{#1..}}
\newcommand{\true}{\ensuremath{\mathit{true}}\xspace}
\newcommand{\false}{\ensuremath{\mathit{false}}\xspace}
\newcommand{\closure}{\ensuremath{\mathit{cl}}}
\newcommand{\Inf}{\mathrm{Inf}}

\DeclareSymbolFont{LTL}{OMS}{pxsy}{m}{n}
\DeclareMathSymbol{\Next}{\mathop}{LTL}{"0D}
\newcommand{\Nextdot}{%
	\mathop{%
		\begin{tikzpicture}[baseline]%
			\node[inner sep=0pt, anchor=base] (n) {$\Next$};%
			\fill (n) circle (.6pt);%
		\end{tikzpicture}%
	}
}
\DeclareMathSymbol{\Until}{\mathbin}{LTL}{"55}
\newcommand{\Untildot}{\ensuremath{\mathrel{\topinset{$\cdot$}{$\Until$}{3pt}{0.6pt}}}\xspace}
\newcommand{\Release}{\ensuremath{\mathrel{\mathcal{R}}}\xspace}
\newcommand{\Releasedot}{\ensuremath{\mathrel{\topinset{$\cdot$}{$\Release$}{1.3pt}{0.9pt}}}\xspace}

\begin{document}
\begin{abstract}
Although it is widely accepted that every system should be robust, in the sense that  ``small'' violations of environment assumptions should lead to ``small'' violations of system guarantees, it is less clear how to make this intuitive notion of robustness mathematically precise. In this paper, we address this problem by developing a robust version of Linear Temporal Logic (LTL), which we call \emph{robust} LTL and denote by rLTL. Formulas in rLTL are syntactically identical to LTL formulas  but are endowed with a many-valued semantics that encodes robustness. In particular, the semantics of the rLTL formula $\varphi\Rightarrow \psi$ is such that a ``small'' violation of the environment assumption $\varphi$ is guaranteed to only produce a ``small'' violation of the system guarantee $\psi$. In addition to introducing rLTL, we study the verification and synthesis problems for this logic: similarly to LTL, we show that both problems are decidable, that the verification problem can be solved in time exponential in the number of subformulas of the rLTL formula at hand, and that the synthesis problem can be solved in doubly exponential time.
\end{abstract}

\title[]{Robust Linear Temporal Logic}

\author[Paulo Tabuada]{Paulo Tabuada}
\address{Department of Electrical Engineering\\ 
University of California at Los Angeles\\
Los Angeles, CA 90095-1594, USA}
\urladdr{http://www.ee.ucla.edu/$\sim$tabuada}
\email{tabuada@ee.ucla.edu}

\author[Daniel Neider]{Daniel Neider}
\address{Department of Electrical Engineering\\ 
University of California at Los Angeles\\
Los Angeles, CA 90095-1594, USA}
\email{neider@ucla.edu}

\maketitle


\section{Introduction}
Specifications for open reactive systems are typically written as an implication
\begin{equation}
\label{Eq:Specification}
\varphi\Rightarrow \psi,
\end{equation}
where $\varphi$ is an environment assumption and $\psi$ is a system guarantee. In Linear Temporal Logic (LTL), this implication is equivalent to $\neg\varphi\lor \psi$. Hence, whenever the assumption $\varphi$ is violated the system can behave arbitrarily. This is clearly inadequate since environment assumptions will \emph{inevitably} be violated. The true environment where the system will be deployed is not completely known at design time and thus cannot be accurately described by the formula $\varphi$. This observation acquires added significance in the context of cyber-physical systems. These are reactive systems interacting with physical environments that are, in many cases, hard to predict and model. To illustrate this point, just consider the problem of modeling all the physical environments where cyber-physical systems, such as modern automobiles, are expected to operate.

We argue that a robust design satisfies the implication in~\eqref{Eq:Specification} in a robust manner (i.e., a ``small'' violation of $\varphi$ results, at most, in a ``small'' violation of $\psi$). To make this intuitive notion of \emph{robustness} mathematically precise, we introduce in this paper a new logic termed \emph{robust Linear Temporal Logic} and simply denoted by rLTL. We do so while being guided by two objectives: first, the syntax of rLTL should be similar to the syntax of LTL in order to make the transition from LTL to rLTL as transparent as possible; second, robustness should be intrinsic to the logic rather than extrinsic (i.e., robustness should not rely on the ability of the designer to provide quantitative information such as ranks, costs, or quantitative interpretations of atomic propositions). This guarantees that verification and synthesis techniques for rLTL are widely applicable as they \emph{only} require an LTL specification. 

The main conceptual question to be addressed when developing the semantics of rLTL is how to give mathematical meaning to ``small'' violations of a formula $\varphi$. Moreover, the answer should not rely on quantitative information provided by the designer, but it should be entirely based on the LTL formula $\varphi$ and its semantics. The approach advocated in this paper can be intuitively explained by regarding LTL formulas of the form $\Box p$, $\Diamond\Box p$, $\Box\Diamond p$, and $\Diamond p$, for an atomic proposition $p$, as requirements on the number of times that $p$ should be satisfied over time. Under this interpretation, and for the formula $\varphi=\Box p$, there is a clear ordering among the possible temporal evolutions of $p$: $p$ being satisfied at every time instant is preferred to $p$ being violated at finitely many time instants which, in turn, is preferred to $p$ being satisfied and violated at infinitely many time instants. The latter case is preferred to $p$ only being satisfied at finitely many time instants and this case is preferred over $p$ being satisfied at no time instant. A semantics that would distinguish between these different five cases would then enable us to state that violating $\Box p$ while satisfying $\Diamond \Box p$ consists of a smaller violation of the formula $\varphi=\Box p$ than violating $\Box p$ while satisfying $\Box\Diamond p$. Making these ideas mathematically rigorous requires a 5-valued semantics that we develop in this paper. Interestingly, the specific interpretation we make of the five different truth values leads to an intuitionistic semantics where negation is dualized and to a corresponding algebraic structure, \emph{da Costa algebras}, that were only very recently investigated~\cite{Priest09}.

\subsection*{Contributions}
The first contribution of this paper is the new logic rLTL that enables reasoning about robustness of LTL specifications. The syntax of rLTL is identical to the syntax of LTL, except that we decorate the temporal operators with a dot so as to easily distinguish between rLTL and LTL. The 5-valued semantics of rLTL is, however,  quite different in many regards. Although only time can tell if the proposed semantics is the \emph{right} one, we provide compelling arguments that it is both \emph{natural} and \emph{useful}. We argue that it is natural by carefully motivating the need for a many-valued semantics and discussing every choice made in defining the proposed 5-valued semantics. Usefulness is argued by providing several examples illustrating how rLTL can be used to reason about robustness. We start in Section~\ref{Sec:ANew} with the fragment of rLTL that only contains the temporal operators always and eventually. This fragment is simpler than full rLTL, yet illustrates most of the technical difficulties encountered with the new semantics. Full rLTL, including the next, release, and until operators, is discussed in Section~\ref{sec:full_rLTL}.

The second contribution is the study of several computational questions related to rLTL. We show that rLTL and LTL are equally expressive by providing effective translations from LTL to rLTL formulas and vice versa. This has two interesting consequences:
\begin{enumerate}
	 \item Any LTL formula can be treated as an rLTL formula (by just dotting the temporal operators), and the LTL semantics can be recovered from the semantics of rLTL. In this way, existing LTL specifications become enriched with a notion of robustness in a completely transparent manner and users do not need to employ a new formalism.
	 \item All (decidability) questions for rLTL are immediately settled.
\end{enumerate}
However, the translation from rLTL to LTL involves an exponential blow-up, thus, leaving open the possibility of improved complexity bounds for the rLTL verification and synthesis problems. Indeed, the exponential blow-up can be avoided by a carefully generalization of the construction that associates with each LTL formula $\varphi$ a Büchi automaton $\mathcal{A}_\varphi$ recognizing all the infinite words satisfying $\varphi$. Critical to this new construction are the properties of the da~Costa algebra, used to define the rLTL semantics, which can be leveraged to keep the size of $\mathcal{A}_\varphi$ in $\mathcal O(|\cl(\varphi)| \cdot 5^{\vert \cl(\varphi)\vert})$ where $\cl(\varphi)$ denotes the set of subformulas of $\varphi$. Note that this is the same complexity bound for LTL where we replace $2$ (since LTL has a $2$-valued semantics) with $5$ (since rLTL has a $5$-valued semantics). Additional consequences of the construction of $\mathcal{A}_\varphi$ include:
\begin{itemize}
	\item the time complexity of verifying rLTL specifications, which we show to be exponential in the size of specification (measured in terms of $\cl(\varphi)$) and polynomial in the size of the system being verified; and
	\item the time complexity of synthesizing reactive controllers for rLTL specifications, which we show to be doubly exponential in the size of the specification and polynomial in the size of the underlying game graph that describes the possible behaviour of an adversarial environment.
\end{itemize}

These results are presented in detail in Section~\ref{Sec:Model} and in Section~\ref{Sec:Quality} we briefly discuss one possible extension of rLTL.

\subsection*{Related efforts}
\label{SSec:Related}
Several efforts to robustify Implication~(\ref{Eq:Specification}) have been reported in the literature. Although most of these efforts started from the same intuitive description of robustness, they resulted in different mathematical formalizations. Bloem et~al.~\cite{BGHJ09} formalized robustness by comparing how often the system violates its assumptions with how often the environment violates its assumptions. Such comparison is performed via a ratio that provides a measure of robustness. Counting the number of violations requires the designer to provide, in addition to the qualitative specification, quantitative information in the form of error functions. In contrast, when working with rLTL, the designer only needs to provide an LTL specification. A very similar approach, based on techniques from robust control, is reported in~\cite{TMD08} where the designer needs to specify maps providing a real-valued interpretation of input and output symbols. A different notion of robustness appeared in the work of Doyen et~al.~\cite{DHLN10}, which requires the effect of a sporadic disturbance to disappear in finite time. If we consider the LTL specification $\Box p\Rightarrow \Box q$ for atomic propositions $p$ and $q$, we can model a sporadic violation of $\Box p$ by $\Diamond\Box p$. The notion of robustness in~\cite{DHLN10} then requires the system to satisfy $\Diamond\Box q$. The semantics of rLTL was built so as to naturally encode this as well as other requirements expressing how a weakening of the system assumptions should lead to a weakening of the system guarantees. Previous work by one of the authors, reported in~\cite{IOSEmsoft,RobustCPS}, provided a single notion of robustness encompassing the notions in~\cite{TMD08} and~\cite{DHLN10} but requiring the designer to provide quantitative information in the form of a cost. Such cost implicitly specifies how guarantees and assumptions are to be weakened in a robust design and was inspired by the work of Alur et~al.~\cite{Ranking} on synthesis for prioritized requirements. A different formalization of robustness appeared in the work of Ehlers and Topcu~\cite{Intermitent}, which considered a specific class of violations of safety assumptions defined by the frequency of violations. In contrast to all the previously described approaches, the results in this paper do not require any additional assumptions or input from a designer beyond an LTL formula. Hence, they apply to any specification that can be written in LTL.

All the previously described approaches addressed safety requirements. In contrast, the work of Bloem~et~al.\ in~\cite{RobustnessLiveness} focused on liveness. The authors considered specifications of the form $\land_{i\in I}\,\varphi_i\Rightarrow \land_{j\in J}\,\psi_j$, where $\varphi_i$ and $\psi_j$ are formulas of the form $\Diamond\Box p$ for some atomic proposition $p$ (depending on $i$ and $j$). Robustness is then measured by comparing the number of violated environment assumptions $\varphi_i$ with the number of violated system guarantees $\psi_j$. This approach is incomparable with ours since the rLTL semantics does not distinguish between the violation of one assumption from the violation of multiple assumptions.\kern-.06em\footnote{In Section~\ref{SSec:daCosta}, we argue why this is desirable and briefly mention how a different semantics for conjunction could be constructed for the purpose of distinguishing between different numbers of assumptions being violated.} It does, however, distinguish between the different ways in which $\varphi_i$ and $\psi_j$ can be violated. Although robustness is formalized differently, rLTL can be used to reason about the robustness of both safety and liveness specifications as long as such properties can be encoded in LTL. Also incomparable with the methods described in this paper is the work of Chaudhuri et~al.~\cite{ChaudhuriGL10} and of Majumdar and Saha~\cite{MS09}, which consider continuity properties of software expressed by the requirement that a deviation in a program's input causes a proportional deviation in its output. Although natural, these notions of robustness only apply to the Turing model of computation and not to the reactive model of computation employed in this paper.

There exists a large body of work on many-valued logics that we will not attempt to review here since it does not directly address questions of robustness. We do, however, allow for two exceptions. The first is the work of Almagor et~al.~\cite{Quality}, which employs a many-valued variant of LTL to reason about quality. The use of a many-valued semantics in the context of quality is as natural as in the context of robustness. In fact, we show in Section~\ref{Sec:Quality} that by dualizing the semantics of rLTL in a specific sense we obtain a logic that is adequate to reason about quality. Nevertheless, there are strong conceptual differences between the approach taken in this paper and the approach in~\cite{Quality}. First, our notion of robustness or quality is intrinsic to the logic, while the approach in~\cite{Quality} requires the designer to provide an interpretation of each atomic proposition in the interval $[0,1]$. Second, there are several choices to define the logical connectives on the interval $[0,1]$. As an illustration for the latter, note that there are three commonly used conjunctions: \L{}ukasiewicz's conjunction $a\land b=\max\{0,a+b-1\}$, Gödel's conjunction $a\land b=\min\{a,b\}$, and the product of real numbers $a\land b=a \cdot b$ also known as Goguen's conjunction. Moreover, each such choice leads to a different notion of implication via residuation. Whether Gödel's conjunction, used in~\cite{Quality}, is the most adequate to formalize quality is a question not addressed in~\cite{Quality}. In contrast, we carefully discuss and motivate all the choices made when defining the semantics of rLTL with robustness considerations. The second exception is the work of Fainekos and Pappas~\cite{Fainekos} on robustness of temporal logic over continuous signals and its extensions (e.g., Donze and Maler~\cite{RobustMaler}). As with the work of Almagor et~al., no discussion of the specific choices made when crafting the many-valued semantics is provided in these papers. Moreover, the results in~\cite{Fainekos} and~\cite{RobustMaler} require continuous-valued signals whereas rLTL is to be used in the more classical setting of discrete-time and finite valued signals (e.g., as provided by transition systems).

The last body of work related to the contents of this paper is the work of Kupferman and co-workers on lattice automata and lattice LTL \cite{DBLP:conf/vmcai/KupfermanL07,DBLP:conf/fossacs/AlmagorK14}.
The syntax of lattice LTL is similar to the syntax of LTL except that atomic propositions assume values on a finite lattice (which has to satisfy further restrictions such as being distributive). Although both lattice LTL as well as rLTL are many-valued logics, lattice LTL derives its many-valued character from the atomic propositions. In contrast, atomic propositions in rLTL are interpreted classically (i.e., they only assume two truth values). Therefore, the many-valued character of rLTL arises from the temporal evolution of the atomic propositions and not from the nature of the atomic propositions or their interpretation. In fact, if we only allow two truth values for the atomic propositions in lattice LTL, as is the case for rLTL, lattice LTL degenerates into LTL. Hence, these two logics capture orthogonal considerations, and results on lattice LTL and lattice automata do not shed light on how to address similar problems for rLTL. 


\section{Notation and Review of Linear Temporal Logic}\label{sec:preliminaries}

Let $\mathbb N = \{0, 1, \ldots\}$ be the set of natural numbers and $\B = \{0, 1\}$ the set of Boolean values with $0$ interpreted as \false and $1$ interpreted as \true. For a set $S$, let $2^S$ be the \emph{powerset} of $S$ and $S^\omega$ the set of all \emph{infinite sequences} of elements of $S$.

An \emph{alphabet}, usually denoted by the Greek letter $\Sigma$, is a finite, nonempty set whose elements are called \emph{symbols}. An infinite sequence $\sigma = a_0 a_1 \ldots$ of  symbols with $a_i \in \Sigma$, $i \in \mathbb N$, is called an \emph{infinite word}.
For an infinite word $\sigma = a_0 a_1\hdots \in \Sigma^\omega$ and $i \in \mathbb N$, let $\sigma(i) = a_i$ denote the $i$-th symbol of $\sigma$ and $\sigma_\suf{i}$ the (infinite) suffix of $\sigma$ starting at position $i$ (i.e., $\sigma_\suf{i} = \sigma_{i} \sigma_{i+1} \ldots \in \Sigma^\omega$). In particular, we have the equality $\sigma_\suf{0}=\sigma$.

\emph{Linear Temporal Logic (LTL)} is parameterized by so-called \emph{atomic propositions}, which form the basic building blocks of LTL formulas. The syntax of LTL is defined as follows.

\begin{definition}[LTL syntax]
Let $\mathcal P$ be a nonempty, finite set of atomic propositions. \emph{LTL formulas} are inductively defined as follows:
\begin{itemize}
	\item each $p \in \mathcal P$ is an LTL formula; and
	\item if $\varphi$ and $\psi$ are LTL formulas, so are $\lnot \varphi$, $\varphi \lor \psi$, $\Next \varphi$, $\Box \varphi$, $\Diamond \varphi$, and $\varphi \Until \psi$.
\end{itemize}
\end{definition}

For notational convenience, we add syntactic sugar and allow the formulas $\true$, $\false$, $\varphi \land \psi$, and $\varphi \Rightarrow \psi$ with their usual meaning (i.e., $\true \coloneqq p \lor \lnot p$ for an arbitrary $p \in \mathcal P$, $\false \coloneqq \lnot \true$, $\varphi \land \psi \coloneqq \lnot (\lnot \varphi \lor \lnot \psi)$, and $\varphi \Rightarrow \psi \coloneqq \lnot \varphi \lor \psi$). 
Note that we consider the operators $\Box$ and $\Diamond$ as part of the syntax although they can be defined using the operator $\Until$. We do this purposefully because it allows us to consider the fragment of LTL containing $\Box$ and $\Diamond$ as the only temporal operators without the need to resort to the operator $\Until$.

Usually, one defines the semantics of LTL in terms of a satisfiability relation that relates an LTL formula over the atomic propositions $\mathcal P$ to infinite words over $\Sigma = 2^\mathcal P$. Perhaps less common, but mathematically equivalent, is to define the semantics by a mapping $W$ that maps an infinite word $\sigma \in \Sigma^\omega$ and an LTL formula $\varphi$ to the element $W(\sigma,\varphi) \in \B$. We follow this approach in Section~\ref{Sec:ANew} when proposing the semantics for rLTL and, for the sake of consistency, we also use this approach for LTL. The formal definition is as follows.

\begin{definition}[LTL semantics]
The \emph{LTL semantics} is a mapping $W$, called \emph{valuation}, that is inductively defined as follows:
\begin{itemize}
	\item $W(\sigma, p) = \begin{cases} 0 & \text{$p \notin \sigma(0)$; and} \\ 1 & \text{$p \in \sigma(0)$.} \end{cases}$
	\item $W(\sigma, \lnot \varphi) = 1 - W(\sigma, \varphi)$.
	\item $W(\sigma, \varphi \lor \psi) = \max{\{ W(\sigma, \varphi), W(\sigma, \psi) \}}$.
	\item $W(\sigma, \Next \varphi) = W(\sigma_\suf{1}, \varphi)$.
	\item $W(\sigma, \Box \varphi) = \inf_{i \geq 0}{W(\sigma_\suf{i}, \varphi)}$.
	\item $W(\sigma, \Diamond \varphi) = \sup_{i \geq 0}{W(\sigma_\suf{i}, \varphi)}$.
	\item $W(\sigma, \varphi \Until \psi) =\sup_{j \geq 0}{\min{\{ W(\sigma_\suf{j}, \psi), \inf_{0 \leq i < j}{W(\sigma_\suf{i}, \varphi)} \}}}$.
\end{itemize}
\end{definition}

We often use a compact notation when referring to infinite words over sets of atomic propositions: instead of writing the set of atomic propositions corresponding to a symbol, we use simple propositional formulas, such as $p$, $\lnot p$, and $p \land q$, to denote all the sets of atomic propositions where these formulas hold true according to the LTL semantics. For instance, given an alphabet $\Sigma = 2^\mathcal P$ over $\mathcal P = \{p, q, r \}$, we write $p$ to denote the sets (symbols) $\{p\}, \{p,q\}, \{p,r\}, \{p,q,r\}\in \Sigma$, we write $\lnot p$ to denote the sets $\emptyset, \{q\}, \{r\}, \{q,r\}\in \Sigma$, and we write $p \land q$ to denote the sets $\{p, q\}, \{p,q,r\}\in \Sigma$.


\section{The Syntax and Semantics of Robust Linear Temporal Logic}
\label{Sec:ANew}
In this section, we consider the fragment of LTL that only allows the temporal operators $\Box$ and $\Diamond$, denoted by LTL$(\Box,\Diamond)$, and develop a robust semantics for this fragment, denoted by rLTL$(\boxdot,\Diamonddot)$.
On the one hand, the fragment rLTL$(\boxdot,\Diamonddot)$ is simple enough that we can provide a lucid intuitive explanation for the proposed semantics. On the other hand, rLTL$(\boxdot,\Diamonddot)$ already illustrates most of the technical difficulties encountered with the new semantics. Although we only discuss the semantics of full rLTL in Section~\ref{sec:full_rLTL}, for the purpose of having a single definition, the syntax of full rLTL is introduced in this section.

\subsection{The Syntax of Robust Linear Temporal Logic}
\label{SSec:SyntaxrLTL}
The syntax of rLTL closely mirrors the syntax of LTL with the only noticeable difference being the use of dotted temporal operators.

\begin{definition}[rLTL syntax]\label{def:rLTL_syntax}
Let $\mathcal P$ be a nonempty, finite set of atomic propositions. \emph{rLTL formulas} are inductively defined as follows:
\begin{itemize}
	\item each $p \in \mathcal P$ is an rLTL formula; and
	\item if $\varphi$ and $\psi$ are rLTL formulas, so are $\lnot \varphi$, $\varphi \lor \psi$, $\varphi \land \psi$, $\varphi \Rightarrow \psi$, $\Nextdot \varphi$, $\boxdot \varphi$, $\Diamonddot \varphi$, $\varphi \Releasedot \psi$, and $\varphi \Untildot \psi$.
\end{itemize}
\end{definition}

In LTL, we can derive the conjunction and implication operators from negation and disjunction. This is no longer the case in rLTL since it has a many-valued semantics. For this reason, we directly included conjunction and implication in Definition~\ref{def:rLTL_syntax}. The same reason justifies the presence of the release operator $\Releasedot$ which, in the case of LTL, can be derived from the until and negation operators as $\varphi\Release \psi=\neg\left(\neg\varphi\Until\psi\right)$.

\subsection{Robustness and counting}
\label{SSec:RobustnessCounting}
Consider the LTL formula $\Box p $ where $p$ is an atomic proposition. There is only one way in which this formula can be satisfied, namely that $p$ holds at every time step. In contrast, there are several ways in which this formula can be violated, and we seek a semantics that distinguishes between these. Such distinction, however, should be limited by what can be expressed in LTL so that we can easily leverage the wealth of existing results on verification of, and synthesis from, LTL specifications.

It seems intuitively clear to the authors that the worst manner in which $\Box p$ fails to be satisfied occurs when $p$ fails to hold at every time step. Although still violating $\Box p$, we would prefer a situation where $p$ holds for at most finitely many time instants. Better yet would be that $p$ holds at infinitely many instants while it fails to hold also at infinitely many instants. Finally, among all the possible ways in which $\Box p$ can be violated, we would prefer the case where $p$ fails to hold for at most finitely many time instants. Consequently, our robust semantics is designed to distinguish between satisfaction and these four possible different ways to violate $\Box p$. However, as convincing as this argument might be, a question persists: in which sense can we regard these five alternatives as canonical?

We answer this question by interpreting satisfaction of $\Box p$ as a counting problem. Recall the LTL semantics of $\Box p$ for a word $\sigma$ given by
\begin{align}
\label{LTLBox}
W(\sigma,\Box p)=\inf_{i\ge 0} W(\sigma_\suf{i},p).
\end{align}
The previously discussed five different cases, satisfaction and four different types of violation, can be seen as the result of counting the number of occurrences of $0$s and $1$s in the infinite word \mbox{$\alpha=W(\sigma_\suf{0},p)W(\sigma_\suf{1},p)\hdots\in \B^\omega$} rather than using the $\inf$-operator in~\eqref{LTLBox}. From this perspective, satisfaction corresponds to the number of occurrences of $0$ being zero. Among all the possible ways in which $\Box p$ can be violated, the most preferred occurs when $p$ only fails to hold at finitely many time instants. This corresponds to having a finite number of $0$s in $\alpha$. The next preferred way in which $\Box p$ can be violated occurs when $p$ holds infinitely many times and also fails to hold infinitely many times. This corresponds to having an infinite number of $0$s and of $1$s in $\alpha$. All the other ways in which $\Box p$ can be violated are similarly identified by counting the number of occurrences of $0$s and $1$s.

We say that an LTL$(\Box, \Diamond)$ formula $\varphi$ is a counting formula if its valuation $W(\sigma,\varphi)$ only depends on the number of occurrences of each  atomic proposition but not on its order. Such formula $\varphi$ is essentially counting how many times each atomic proposition appears along the word $\sigma$. Formally, we say that $\varphi$ is a counting formula if for every infinite word $\sigma \in\Sigma^\omega$, seen as a map $\sigma \colon \N \to \Sigma$, and for every bijection $f \colon \N \to \N$ we have
\[ W(\sigma,\varphi)=W(\sigma\circ f,\varphi). \]
Recall that by composing a sequence of permutations (bijections) one again obtains a bijection. Hence, by permuting the elements of $\sigma$, we obtain the word $\sigma\circ f$ where $f$ is the composition of the employed permutations. If we now assume $\mathcal{P}=\{p\}$, then we can always permute the elements of $\sigma$ so that the permuted word $\sigma\circ f$ is of the form 
\[ (\neg p\, p )^k p^\omega, \quad \left(\neg p\, p\right)^\omega,\quad \text{or } \left(\neg p\, p\right)^k \left(\neg p\right)^\omega, \]
where $k\in \N$.

We further recall that formulas in LTL$(\Box, \Diamond)$ can only define stutter-invariant properties~\cite{StutterInvariance}. Therefore, the semantics of LTL$(\Box, \Diamond)$ cannot distinguish\footnote{To see why this is the case, note that any word $\left(\neg p\, p\right)^k p^\omega$ with $k\in\N$ can be permuted to the form $\left(\neg p\right)^k p^k p^\omega$ and by stutter invariance can be reduced to $\neg p\, p\, p^{\omega}$.} between the words $\left(\neg p\, p\right)^{k_1}p^\omega$ and $\left(\neg p\, p\right)^{k_2}p^\omega$ for $k_1\ne k_2$, and $k_1,k_2>0$, although it can distinguish between the case $k_1=0$ and $k_2>0$. The same argument applies to the words $\left(\neg p\, p\right)^{k_1}\left(\neg p\right)^\omega$ and $\left(\neg p\, p\right)^{k_2}\left(\neg p\right)^\omega$ and shows that there are only five canonical forms that can be distinguished by LTL$(\Box, \Diamond)$:
\begin{equation}
\label{Eq:CanonicalForms}
p^\omega,\quad \left(\neg p\, p\right)^+p^\omega,\quad \left(\neg p\, p\right)^\omega,\quad \left(\neg p\, p\right)^+ \left(\neg p\right)^\omega,\quad \text{and } \left(\neg p\right)^\omega.
\end{equation}
It should be no surprise that these are exactly the five cases we previously discussed. In Section~\ref{SSec:FullrLTL}, when discussing full rLTL, we provide further arguments justifying why these 5 different cases can be seen as canonical.

The considerations in this section suggest the need for a semantics that is 5-valued rather than 2-valued so that we can distinguish between the aforementioned five cases. Therefore, we need to replace Boolean algebras by a different type of algebraic structure that can accommodate a 5-valued semantics. \emph{Da Costa algebras}, reviewed in the next section, are an example of such algebraic structures.

\subsection{da~Costa Algebras}
\label{SSec:daCosta}
According to our motivating example $\Box p$, the desired semantics should have one truth value corresponding to \true and four truth values corresponding to different shades of \false. It is instructive to think of truth values as the elements of $\B^4$ (i.e., the four-fold Cartesian product of $B$) that arise as the possible values of the $4$-tuple of LTL formulas:
\begin{equation}
\label{Eq:4TupleLTL}
(\Box p,\Diamond\Box p, \Box\Diamond p, \Diamond p).
\end{equation}
To ease notation, we denote such values interchangeably by $b = b_1 b_2 b_3 b_4$ and $b = (b_1, b_2, b_3, b_4)$ with $b_i\in \B$ for $i\in\{1,2,3,4\}$. The value 1111 then corresponds to \true since $\Box p$ is satisfied. The most preferred violation of $\Box p$ ($p$ fails to hold at only finitely many time instants) corresponds to $0111$, followed by $0011$ ($p$ holds at infinitely many instants and also fails to hold at infinitely many instants), $0001$ ($p$ holds at most at finitely many instants), and $0000$ ($p$ fails to hold at every time instant). Such preferences can be encoded in the linear order
\begin{equation}
\label{Eq:Order}
0000\prec 0001\prec 0011\prec 0111\prec 1111
\end{equation}
that renders the set
\[\B_4=\{0000,0001,0011,0111,1111\}\]
a (bounded) distributive lattice with top element $\top=1111$ and bottom element $\bot=0000$. Formally, $\B_4$ is the subset of $\B^4$ consisting of the $4$-tuples $(b_1,b_2,b_3,b_4)\in \B^4$ satisfying the monotonicity property
\begin{align}
\label{Monotonicity}
i\leq_\N j\text{ implies }b_i\leq_\B b_j
\end{align}
where $i,j\in\{1,\ldots,4\}$, $\leq_\N$ is the natural order on the natural numbers, and $\leq_\B$ is the natural order on the Boolean algebra $\B$. In $\B_4$, the meet $\sqcap$ can be interpreted as minimum and the join $\sqcup$ as maximum with respect to the order in~\eqref{Eq:Order}. We use $\sqcap$ and $\sqcup$ when discussing  lattices in general and use $\min$ and $\max$ for the specific lattice $\B_4$ or the Boolean algebra $\B$.

The first choice to be made in using the lattice $(\B_4,\min,\max)$ to define the semantics of rLTL$(\boxdot,\Diamonddot)$ is the choice of an operation on $\B_4$ modeling conjunction. It is well know that all the desirable properties of a many-valued conjunction are summarized by the notion of triangular-norm, see~\cite{MetamatematicsFuzzyLogic,PrinciplesFuzzyLogic}. One can compare two triangular-norms $s$ and $t$ using the partial order defined by declaring $s\le t$ when $s(a,b)\le t(a,b)$ for all $a,b\in \B_4$. According to this order, the triangular-norm $\min$ is maximal among all triangular-norms (i.e., we have $t(a,b)\le \min\{a,b\}$ for every $a,b\in \B_4$ and every triangular-norm $t$). This shows that if we choose any triangular-norm $t$ different from $\min$, there exist elements $a,b\in \B_4$ for which we have $t(a,b)<\min\{a,b\}$. Hence, any choice different from $\min$ would result in situations where the value of a conjunction is \emph{smaller} than the value of the conjuncts, which is not reasonable when interpreting the value of the conjuncts as different shades of \false. To illustrate this point, consider the formula $\Box p\land \Box q$ and the word $\sigma=\neg(p\land q)(p\land q)^\omega$. As introduced above, the value of $\Box p$ on $\sigma$ corresponds to $0111$ and the value of $\Box q$ on $\sigma$ corresponds to $0111$ since on both cases we have the most preferred violation of the formulas. Therefore, the value of $\Box p\land \Box q$ on $\sigma$ should also be $0111$ since the formula $\Box p\land \Box q$ is only violated a finite number of times. It thus seems natural\footnote{Note that there are situations where it is convenient to model conjunction differently. In Section~\ref{SSec:Related}, we referenced the work of Bloem et~al.~\cite{RobustnessLiveness}, where the specific way in which robustness is modeled requires distinguishing between the number of conjuncts that are satisfied in the assumption $\land_{i\in I} \varphi_i$. This cannot be accomplished if conjunction is modeled by $\min$ and a different triangular-norm would have to be used for this purpose. Note that both \L{}ukasiewicz's conjunction as well as Goguen's conjunction, briefly mentioned in Section~\ref{SSec:Related}, have the property that their value decreases as the number of conjuncts that are true decreases.} to model conjunction in $\B_4$ by $\min$ and, for similar reasons, to model disjunction in $\B_4$ by $\max$. 

As in intuitionistic logic\footnote{This is also done in context of residuated lattices that is more general than the Heyting algebras used in intuitionistic logic. Recall that a residuated lattice is a lattice $(A,\sqcap,\sqcup)$, satisfying same additional conditions, and equipped with a commutative monoid $(A,\otimes,\textbf{1})$ satisfying some additional compatibility conditions. Since we chose the lattice meet $\sqcap$ to represent conjunction, we have a residuated lattice where $\otimes=\sqcap$ and $\textbf{1}=\top$.}, our implication is defined as the residue of $\sqcap$. In other words, we define the implication $a\rightarrow b$ by requiring that $c \preceq a \rightarrow b$ if and only if $c \sqcap a \preceq b$ for every $c\in \mathbb{B}_4$. This leads to
\[ a\rightarrow b= \begin{cases} 1111 & \text{if $a\preceq b$; and} \\ b & \text{otherwise.} \end{cases} \]

However, we now \emph{diverge} from intuitionistic logic (and most many-valued logics) where negation of $a$ is defined by $a\rightarrow 0000$. Such negation is not compatible with the interpretation that all the elements of $\mathbb{B}_4$, except for $1111$, represent (different shades of) \false and thus their negation should have the truth value $1111$. To make this point clear, we present in Table~\ref{Neg} the intuitionistic negation in $\B_4$ and the desired negation compatible with the  interpretation of the truth values in $\B_4$.  

\begin{table}[ht]
	\centering
	\caption{Desired negation vs.\ intuitionistic negation in $\B_4$.}\label{Neg}
	\begin{tabular}{c@{\hskip 3em}cc}	
		\toprule
		& Desired & Intuitionistic \\
		Value & negation & negation\\
		\midrule
		1111 & 0000 & 0000\\
		0111 & 1111 & 0000\\
		0011 & 1111 & 0000\\
		0001 & 1111 & 0000\\
		0000 & 1111 & 1111\\
		\bottomrule
	\end{tabular}
\end{table}

What is then the algebraic structure on $\B_4$ that supports the desired negation, dual to the intuitionistic negation? This very same problem was recently investigated by Priest~\cite{Priest09} and the answer is \emph{da~Costa} algebras. 

\begin{definition}[da~Costa algebra]
\label{Def:daCostaAlgebra}
A \emph{da Costa} algebra is a $6$-tuple $(A,\sqcap,\sqcup,\preceq,\rightarrow,\overline{\,\cdot\,})$ where
\begin{enumerate} 
\item $(A,\sqcap,\sqcup,\preceq)$ is a distributive lattice where $\preceq$ is the ordering relation derived from $\sqcap$ and $\sqcup$;
\item $\rightarrow$ is the residual of $\sqcap$ (i.e., $a\preceq b\rightarrow c$ if and only if $a\sqcap b\preceq c$ for every $a,b,c\in A$);
\item $a \preceq b\sqcup\overline{b}$ for every $a,b\in A$; and
\item $\overline{a} \preceq b$ whenever $c\sqcup\overline{c}\preceq a\sqcup b$ for every $a,b,c\in A$.
\end{enumerate}
\end{definition}

In a da~Costa algebra, one can define the top element $\top$ to be $\top=a\sqcup\overline{a}$ for an arbitrary $a\in A$; note that $\top$ is unique and independent of the choice of $a$. Hence, the third requirement in Definition~\ref{Def:daCostaAlgebra} amounts to the definition of top element, while the fourth requirement can be simplified to
\[\overline{a} \preceq b\text{ whenever }\top\preceq a\sqcup b.\]

We can easily verify that $\B_4$ is a da Costa algebra if we use the desired negation defined in Table~\ref{Neg}. 

It should be mentioned that working with a $5$-valued semantics has its price. The law of non-contradiction fails in $\B_4$ (i.e., $a \sqcap \overline{a}$ may not equal $\bot=0000$ as evidenced by taking $a=0111$). However, since $a \sqcap \overline{a}\prec 1111$, a weak form of non-contradiction still holds as $a \sqcap \overline{a}$ is to be interpreted as a shade of \false but not necessarily as the least preferred way of violating $a \sqcap \overline{a}$, which corresponds to $\bot$. Contrary to intuitionistic logic, the law of excluded middle is valid (i.e., $a \sqcup \overline{a}=\top=1111$). Finally, $a=0111$ shows that $\overline{\overline{a}}\ne a$ although it is still true that $\overline{\overline{a}}\rightarrow a$. Interestingly, we can think of double negation
\[ \overline{\overline{a}} = \begin{cases} 1111 & \text{if $a=1111$; and} \\ 0000 & \text{otherwise} \end{cases} \]
as quantization in the sense that \true is mapped to \true and all the shades of \false are mapped to \false. Hence, double negation quantizes the five different truth values into two truth values (\true and \false) in a manner that is compatible with our interpretation of truth values.

\subsection{Semantics of rLTL$(\boxdot,\Diamonddot)$ on da~Costa Algebras}
\label{SSec:Semantics}
The semantics of rLTL$(\boxdot,\Diamonddot)$ is given by a mapping $V$, called valuation as in the case of LTL, that maps an infinite word $\sigma\in \Sigma^\omega$ and an rLTL$(\boxdot,\Diamonddot)$ formula $\varphi$ to an element of $\mathbb{B}_4$. In defining $V$, we judiciously use the algebraic operations of the da Costa algebra $\B_4$ to give meaning to the logical connectives in the syntax of rLTL$(\boxdot,\Diamonddot)$. In the following, let $\Sigma=2^\mathcal{P}$ for a finite set of atomic propositions $\mathcal{P}$. 

On atomic propositions $p\in \mathcal{P}$, $V$ is defined by
\begin{align}
V(\sigma, p) = \begin{cases}
0000 & \text{if $p \notin \sigma(0)$; and} \\
1111 & \text{if $p \in \sigma(0)$.}
\end{cases}
\end{align}
Hence, atomic propositions are interpreted classically (i.e., only two truth values are used).
Since we are using a $5$-valued semantics, we provide a separate definition for all the four logical connectives:

\begin{align}
V(\sigma,\varphi\land\psi) &= V(\sigma,\varphi)\sqcap V(\sigma,\psi), \\
V(\sigma,\varphi\lor\psi) &= V(\sigma,\varphi)\sqcup V(\sigma,\psi), \\
V(\sigma,\neg\varphi) &= \overline{V(\sigma,\varphi)}, \\
V(\sigma,\varphi\Rightarrow\psi) &= V(\sigma,\varphi)\rightarrow V(\sigma,\psi).
\end{align}
Note how the semantics mirrors the algebraic structure of da Costa algebras. This is no accident since valuations are typically algebra homomorphisms.

Unfortunately, da Costa algebras are not equipped\footnote{One could consider developing a notion of da Costa algebras with operators in the spirit of Boolean algebras with operators~\cite{BAO}. We leave such investigation for future work.} with operations corresponding to $\boxdot$ and $\Diamonddot$, the robust versions of $\Box$ and $\Diamond$, respectively. Therefore, we resort to the counting interpretation in Section~\ref{SSec:RobustnessCounting} to motivate the semantics of $\boxdot$. Formally, the semantics of  $\boxdot$ is given by
\begin{align}
\label{BoxSemantics}
V(\sigma,\boxdot\varphi)=\left(\inf_{i\ge 0}V_1(\sigma_\suf{i},\varphi),~ \sup_{j\ge 0}\inf_{i\ge j}V_2(\sigma_\suf{i},\varphi),~ \inf_{j\ge 0}\sup_{i\ge j}V_3(\sigma_\suf{i},\varphi),~ \sup_{i\ge 0}V_4(\sigma_\suf{i},\varphi)\right)
\end{align}
where $V_k(\sigma,\varphi)=\pi_k\circ V(\sigma,\varphi)$ for $k\in \{1,2,3,4\}$ and $\pi_k:\B_4\to \B$ are the mappings defined by
\begin{equation}
\label{Eq:Projection} 
\pi_k(a_1,a_2,a_3,a_4)=a_k.
\end{equation}

To illustrate the semantics of $\boxdot$, let us consider the simple case where $\varphi$ is just an atomic proposition $p$. This means that one can express $V(\sigma,\boxdot p)$ in terms of the LTL valuation $W$ by 
\begin{align}
\label{Decompose}
V(\sigma,\boxdot p)=\left(W(\sigma,\Box p),W(\sigma,\Diamond\Box p),W(\sigma,\Box\Diamond p),W(\sigma,\Diamond p)\right).
\end{align}
In other words, $V_1(\sigma,\boxdot p)$ corresponds to the LTL truth value of $\Box p$, $V_2(\sigma, \boxdot p)$ corresponds to the LTL truth value of $\Diamond\Box p$, $V_3(\sigma, \boxdot p)$ corresponds to the LTL truth value of $\Box\Diamond p$, and $V_4(\sigma, \boxdot p)$ corresponds to the LTL truth value of $\Diamond p$. Equation~\eqref{Decompose} connects the semantics of $\boxdot$ to the counting problems described in Section~\ref{SSec:RobustnessCounting} and to the $4$-tuple of LTL formulas in~\eqref{Eq:4TupleLTL}. In Section~\ref{SSec:FullrLTL} we re-interpret Equality~\eqref{Decompose} in the more general context of arbitrary formulas $\varphi$ and full rLTL.

The last operator is $\Diamonddot$, whose semantics is given by
\begin{align}
V(\sigma,\Diamonddot\varphi)=\left(\sup_{i\ge 0}V_1(\sigma_\suf{i},\varphi),~ \sup_{i\ge 0}V_2(\sigma_\suf{i},\varphi),~ \sup_{i\ge 0}V_3(\sigma_\suf{i},\varphi),~ \sup_{i\ge 0}V_4(\sigma_\suf{i},\varphi)\right).
\end{align}
According to the counting problems used in Section~\ref{SSec:RobustnessCounting} to motivate the proposed semantics, there is only one way in which the LTL formula $\Diamond p$, for an atomic proposition $p$, can be violated. Hence, $V(\sigma,\Diamonddot\varphi)$ is one of only two possible truth values: 1111 or 0000. We further note that $\Diamonddot$ is not dual to $\boxdot$, as expected in a many-valued logic where the law of double negation fails.

Having defined the semantics of rLTL$(\boxdot,\Diamonddot)$, let us now see if the formula $\boxdot p\Rightarrow \boxdot q$, where $\boxdot p$ is an environment assumption and $\boxdot q$ is a system guarantee with $p,q\in\mathcal{P}$, lives to the expectations set in the introduction and to the intuition provided in Section~\ref{SSec:RobustnessCounting}.

\begin{enumerate}
\item According to~\eqref{Decompose}, if $\Box p$ holds, then $\boxdot p$ evaluates to $1111$ and the implication $\boxdot p\Rightarrow \boxdot q$ is \true (i.e., the value of $\boxdot p\Rightarrow \boxdot q$ is 1111) if $\boxdot q$ evaluates to $1111$ (i.e., if $\Box q$ holds). Therefore, the desired behavior of $\Box p\Rightarrow \Box q$, when the environment assumptions hold, is retained.
\item Consider now the case where $\Box p$ fails but the weaker assumption $\Diamond\Box p$ holds. In this case $\boxdot p$ evaluates to $0111$ and the implication $\boxdot p \Rightarrow \boxdot q$ is \true if $\boxdot p$ evaluates to $0111$ or higher. This means that $\Diamond\Box q$ needs to hold. 
\item A similar argument shows that we can also conclude the following consequences whenever $\boxdot p\Rightarrow \boxdot q$ evaluates to $1111$: $\Box\Diamond q$ follows whenever the environment satisfies $\Box\Diamond p$ and $\Diamond q$ follows whenever the environment satisfies $\Diamond p$.
\end{enumerate}
We thus conclude that the semantics of $\boxdot p\Rightarrow \boxdot q$ captures the desired robustness property by which a weakening of the assumption $\boxdot p$ leads to a weakening of the guarantee $\boxdot q$. The following examples further motivate the usefulness of the proposed semantics. Additional arguments in favor of the proposed definition of $\boxdot$ and $\Diamonddot$ are given in Section~\ref{SSec:FullrLTL} when defining full rLTL.

\subsection{Examples}

\subsubsection{The usefulness of implications that are not \textbf{true}}
We argued in the previous section that rLTL$(\boxdot,\Diamonddot)$ captures the intended robustness properties for the specification $\boxdot p\Rightarrow \boxdot q$ whenever this formula evaluates to $1111$. But does the formula $\boxdot p\Rightarrow \boxdot q$ still provide useful information when its value is lower than $1111$? It follows from the semantics of implication that $V(\sigma, \boxdot p\Rightarrow \boxdot q)=b$, for $b\prec 1111$, occurs when $V(\sigma,\boxdot q)=b$ (i.e., whenever a value of $b$ can be guaranteed despite $b$ being smaller than $V(\sigma,\boxdot p)$). The value $V(\sigma, \boxdot p\Rightarrow \boxdot q)$ thus describes which weakened guarantee follows from the environment assumption whenever the intended system guarantee does not. This can be seen as another measure of robustness: despite $\boxdot q$ not following from $\boxdot p$, the behavior of the system is not arbitrary, a value of $b$ is still guaranteed.

\subsubsection{GR(1) in rLTL$(\boxdot,\Diamonddot)$}
The GR(1) fragment of LTL is becoming increasingly popular for striking an interesting balance between its expressiveness and the complexity of the corresponding synthesis problem~\cite{GR1}.  Recall that a GR(1) formula is an LTL$(\Box,\Diamond)$ formula of the form
\begin{equation}
\label{GR(1)}
\bigwedge_{i\in I}\Box\Diamond p_i\Rightarrow \bigwedge_{j\in J}\Box\Diamond q_j
\end{equation}
where $p_i$ and $q_j$ are atomic propositions and $I, J$ are finite sets. We obtain the rLTL$(\boxdot,\Diamonddot)$ version of~\eqref{GR(1)} simply by dotting the boxes and the diamonds: 
\begin{equation}
\label{rGR(1)}
\bigwedge_{i\in I}\boxdot\Diamonddot p_i\Rightarrow \bigwedge_{j\in J}\boxdot\Diamonddot q_j.
\end{equation}
Any valuation $V$ for $\boxdot\Diamonddot p_i$ can be expressed in terms of a  valuation $W$ for LTL as
\begin{align*}
V(\sigma, \boxdot\Diamonddot p_i) & =\left(W(\sigma, \Box\Diamond p_i),~ W(\sigma, \Diamond\Box\Diamond p_i),W(\sigma, \Box\Diamond\Diamond p_i),~ W(\sigma, \Diamond\Diamond p_i)\right)\\
& =\left(W(\sigma, \Box\Diamond p_i),~ W(\sigma, \Box\Diamond p_i),~ W(\sigma, \Box\Diamond p_i),~ W(\sigma, \Diamond p_i)\right).
\end{align*}
Therefore, $V(\sigma, \boxdot\Diamonddot p_i)$ can only assume three different values: $1111$ when $\Box\Diamond p_i$ holds, $0001$ when $\Box\Diamond p_i$ fails to hold but $\Diamond p_i$ does hold, and $0000$ when $\Diamond p_i$ fails to hold. Based on this observation, and assuming that~\eqref{rGR(1)} evaluates to $1111$, we conclude that $\bigwedge_{j\in J}\Box\Diamond q_j$ holds whenever $\bigwedge_{i\in I}\Box\Diamond p_i$ does, as required by~\eqref{GR(1)}. In contrast with~\eqref{GR(1)}, however, the weakened system guarantee $\bigwedge_{j\in J}\Diamond q_j$ holds whenever the weaker environment assumption $\bigwedge_{i\in I}\Diamond p_i$ does.

\subsubsection{Non-counting formulas}
All the preceding examples were counting formulas, as defined in Section~\ref{SSec:RobustnessCounting}. We now consider the simple non-counting formula $\Box(p\Rightarrow \Diamond q)$, which requires each occurrence of $p$ to be followed by an occurrence of $q$. The word $(p\land \neg q)(\neg p\land q)(\neg p\land \neg q)^\omega$ clearly satisfies this formula although its permutation $(\neg p\land q)(p\land \neg q)(\neg p\land \neg q)^\omega$ does not. In addition to being a non-counting formula, $\Box(p\Rightarrow \Diamond q)$ is one of the most popular examples of an LTL formula used in the literature and, for this reason, constitutes a litmus test to rLTL$(\boxdot,\Diamonddot)$. The semantics of the dotted version of $\Box (p\Rightarrow \Diamond q)$ can be expressed using an LTL valuation $W$ as
\[V(\sigma,\boxdot(p\Rightarrow \Diamonddot q))=\left(W(\sigma,\Box(p\Rightarrow \Diamond q)),~ W(\sigma,\Box\Diamond p\Rightarrow \Box\Diamond q),~ W(\sigma,\Diamond\Box p\Rightarrow \Box\Diamond q),~ W(\sigma,\Box p\Rightarrow \Diamond q)\right).\]
It is interesting to observe how the semantics of $\varphi=\boxdot(p\Rightarrow\Diamonddot q)$ recovers: strong fairness, also known as compassion, when the value of $\varphi$ is 0111; weak fairness, also known as justice, when the value of $\varphi$ is 0011; and the even weaker notion of fairness represented by the LTL formula $\Box p\Rightarrow \Diamond q$, when the value of $\varphi$ is 0001. The fact that all these different and well known notions of fairness naturally appear in the proposed semantics is another strong indication of rLTL's naturalness and usefulness.

\subsection{Relating LTL$(\Box,\Diamond)$ and rLTL$(\boxdot,\Diamonddot)$}
\label{SSec:Relating}
In this section we discuss, at the technical level, the relationships between rLTL$(\boxdot,\Diamonddot)$  and LTL$(\Box,\Diamond)$. 

Recall the mapping $\pi_1:\B_4\to\B$ introduced in~\eqref{Eq:Projection}, defined by $\pi_1(a_1,a_2,a_3,a_4)=a_1$. Composing $\pi_1$ with a valuation $V$ of rLTL$(\boxdot,\Diamonddot)$ we obtain the function $V_1=\pi_1\circ V$ transforming an infinite word $\sigma\in \Sigma^\omega$ and a rLTL$(\boxdot,\Diamonddot)$ formula $\varphi$ into the element $V_1(\sigma,\varphi)$ of $\mathbb{B}$. We now show that $V_1$ is in fact a LTL$(\Box,\Diamond)$ valuation.

On atomic propositions $p\in\mathcal{P}$ we have
\begin{align}
V_1(\sigma, p) = \begin{cases}
\pi_1(0000)=0 & \text{if $p \notin \sigma(0)$; and} \\
\pi_1(1111)=1 & \text{if $p \in \sigma(0)$.}
\end{cases}
\end{align}
Moreover, the following equalities can be easily verified:
\begin{align}
\label{Eq:MorFirst}
V_1(\sigma,\varphi\land\psi) & = \pi_1\left(V(\sigma,\varphi)\sqcap V(\sigma,\psi)\right) = \min\{V_1(\sigma,\varphi), V_1(\sigma,\psi)\}, \\
V_1(\sigma,\varphi\lor\psi) & = \pi_1\left(V(\sigma,\varphi)\sqcup V(\sigma,\psi)\right) = \max\{V_1(\sigma,\varphi), V_1(\sigma,\psi)\}, \\
V_1(\sigma,\neg\varphi) & = \pi_1\left(\overline{V(\sigma,\varphi)}\right)=1-\pi_1\left(V(\sigma,\varphi)\right)=1-V_1(\sigma,\varphi), \\
V_1(\sigma,\varphi\Rightarrow\psi) &= \pi_1\left(V(\sigma,\varphi)\rightarrow V(\sigma,\psi)\right)=\max\left\{1-V_1(\sigma,\varphi),V_1(\sigma,\psi)\right\}.
\end{align}
Finally, it follows directly from the semantics of $\boxdot$ and $\Diamonddot$ that
\begin{align}
V_1(\sigma,\boxdot\varphi) & = \pi_1\left(V(\sigma,\boxdot\varphi)\right)=\inf_{i\ge 0}V_1(\sigma_\suf{i},\varphi), \\
\label{Eq:MorLast}
V_1(\sigma,\Diamonddot\varphi) & = \pi_1\left(V(\sigma,\Diamonddot\varphi)\right)=\sup_{i\ge 0}V_1(\sigma_\suf{i},\varphi).
\end{align}

Hence, the semantics of LTL$(\Box,\Diamond)$ can always be recovered from the first component of the semantics of rLTL$(\boxdot,\Diamonddot)$, thereby showing that rLTL$(\boxdot,\Diamonddot)$ is as expressive as LTL$(\Box,\Diamond)$. 

Conversely, one can translate an rLTL$(\boxdot,\Diamonddot)$ formula $\varphi$ into four LTL$(\Box,\Diamond)$ formulas $\psi_\varphi^1, \ldots, \psi_\varphi^4$ such that
\[ \pi_j(V(\sigma, \varphi)) =V_j(\sigma, \varphi) = W(\sigma, \psi_\varphi^j) \]
for all $\sigma \in \Sigma^\omega$ and $j \in \{1, \ldots, 4\}$. The key idea is to emulate the semantics of each operator occurring in $\varphi$ component-wise by means of dedicated LTL formulas.

The construction of $\psi_\varphi^j$ proceeds by induction over the subformulas of $\varphi$: 

\begin{itemize}
	\item If $\varphi = p$ for an atomic proposition $p \in \mathcal P$, then $\psi_\varphi^j \coloneqq p$ for all $j \in \{ 1, \ldots, 4 \}$.
	\item If $\varphi = \varphi_1 \lor \varphi_2$, then $\psi_\varphi^j \coloneqq \psi_{\varphi_1}^j \lor \psi_{\varphi_2}^j$ for all $j \in \{ 1, \ldots, 4 \}$.
	\item If $\varphi = \varphi_1 \land \varphi_2$, then $\psi_\varphi^j \coloneqq \psi_{\varphi_1}^j \land \psi_{\varphi_2}^j$ for all $j \in \{ 1, \ldots, 4 \}$.
	\item If $\varphi = \lnot \varphi_1$, then $\psi_\varphi^j \coloneqq \lnot (\psi_{\varphi_1}^1 \land \psi_{\varphi_1}^2 \land \psi_{\varphi_1}^3 \land \psi_{\varphi_1}^4)$ for all $j \in \{ 1, \ldots, 4 \}$.
	\item If $\varphi = \varphi_1 \Rightarrow \varphi_2$, then $\psi_\varphi^j \coloneqq \left( \bigvee_{k \in \{1, \ldots, 4\}} ( \psi_{\varphi_1}^k \land \lnot \psi_{\varphi_2}^k) \right) \Rightarrow \psi_{\varphi_2}^j$ for all $j \in \{ 1, \ldots, 4 \}$.
	%
	\item If $\varphi = \Diamonddot \varphi_1$, then $\psi_\varphi^j \coloneqq \Diamond \psi_{\varphi_1}^j$ for all $j \in \{ 1, \ldots, 4 \}$.
	\item If $\varphi = \boxdot \varphi_1$, then $\psi_\varphi^1 \coloneqq \Box \psi_{\varphi_1}^1$, $\psi_\varphi^2 \coloneqq \Diamond \Box \psi_{\varphi_1}^2$, $\psi_\varphi^3 \coloneqq \Box \Diamond \psi_{\varphi_1}^3$, and $\psi_\varphi^4 \coloneqq \Diamond \psi_{\varphi_1}^4$.
\end{itemize}

It is not hard to verify that the formulas $\psi_\varphi^j$ have indeed the desired meaning. However, note that the size of $\psi_\varphi^j$, measured in the number of  subformulas, is exponential in the size of $\varphi$ due to the recursive substitution of the sub-formulas.

The preceding discussion can be summarized by following result.

\begin{proposition}
LTL$(\Box,\Diamond)$ and rLTL$(\boxdot,\Diamonddot)$ are equally expressive.
\end{proposition}

Since the translations from LTL$(\Box,\Diamond)$ to rLTL$(\boxdot,\Diamonddot)$ and vice versa are effective, we immediately conclude that any problem for rLTL$(\boxdot,\Diamonddot)$, whose corresponding problem for LTL$(\Box,\Diamond)$ is decidable, is also decidable. In practice, however, the translation from rLTL$(\boxdot,\Diamonddot)$ to LTL$(\Box,\Diamond)$ involves an exponential blow-up. Hence, we investigate in Section~\ref{Sec:Model} the complexity of several verification and synthesis problems by developing algorithms specialized for rLTL$(\boxdot,\Diamonddot)$.


\section{Model Checking and Synthesis}
\label{Sec:Model}
Similarly to LTL, rLTL gives rise to various (decision) problems, some of which we investigate in this section.
We are particularly interested in model checking and in reactive synthesis.
These two problems are clearly amongst the most important in the context of LTL and, hence, must be investigated for rLTL. We address in this section the fragment rLTL$(\boxdot, \Diamonddot)$ and leave full rLTL to Section~\ref{sec:full_rLTL} since this more general case can be handled by a simple extension of the ideas developed for rLTL$(\boxdot, \Diamonddot)$.

As the translation from rLTL$(\boxdot, \Diamonddot)$ into LTL$(\Box, \Diamond)$ potentially results in an exponentially large formula, we now develop a computationally more efficient approach to the model checking and reactive synthesis problems via a translation into (generalized) Büchi automata. Our construction follows the well known translation of LTL into Büchi automata (see, e.g., Baier and Katoen~\cite{Baier:2008:PMC:1373322}) and results in a generalized Büchi automaton with $\mathcal O(k \cdot 5^k)$ states where $k$ counts the subformulas of the given rLTL$(\boxdot, \Diamonddot)$ formula.
This is the same complexity as for the LTL translation---which results in an automation with size in $\mathcal O(k \cdot 2^k)$---once we replace $2$ with $5$ since rLTL is 5-valued while LTL is 2-valued.

Similarly to LTL, our translation relies on so-called expansion rules, which we introduce in Section~\ref{sec:LTL_expansion_rules}. Based on these rules, we present the translation from rLTL$(\boxdot, \Diamonddot)$ to generalized Büchi automata in Section~\ref{sec:rLTL_to_Buechi}.
Subsequently, we consider model checking in Section~\ref{sec:model_checking} and reactive synthesis in Section~\ref{sec:synthesis}.

\subsection{Expansion Rules}\label{sec:LTL_expansion_rules}
The operators $\boxdot$ and $\Diamonddot$ have expansion rules similar to their LTL counterparts $\Box$ and $\Diamond$ (see Baier and Katoen~\cite{Baier:2008:PMC:1373322} for a more in-depth discussion of LTL expansion rules). The following proposition states these rules in detail.

\begin{proposition}[Expansion Rules]\label{prop:expansion_rules}
For any rLTL$(\boxdot,\Diamonddot)$ formula $\varphi$, any $\sigma\in \Sigma^\omega$, any $\ell\in\N$, and any valuation $V$, the following equalities (called \emph{expansion rules}) hold:
\begin{align}
\label{Eq:Exp1}
V_1(\sigma_\suf{\ell},\boxdot\varphi)&=\min\left\{V_1(\sigma_\suf{\ell},\varphi),V_1(\sigma_\suf{\ell+1},\boxdot\varphi)\right\}, \\
\label{Eq:Exp2}
V_2(\sigma_\suf{\ell},\boxdot\varphi) & =\max\left\{V_1(\sigma_\suf{\ell},\boxdot\varphi),V_2(\sigma_\suf{\ell+1},\boxdot\varphi)\right\}, \\
\label{Eq:Exp3}
V_3(\sigma_\suf{\ell},\boxdot\varphi) & =\min\left\{V_4(\sigma_\suf{\ell},\boxdot\varphi),V_3(\sigma_\suf{\ell+1},\boxdot\varphi)\right\}, \\
\label{Eq:Exp4}
V_4(\sigma_\suf{\ell},\boxdot\varphi) &= \max\left\{V_4(\sigma_\suf{\ell},\varphi),V_4(\sigma_\suf{\ell+1},\boxdot\varphi)\right\}, \\
V_k(\sigma_\suf{\ell},\Diamonddot\varphi) &= \max\left\{V_k(\sigma_\suf{\ell},\varphi),V_k(\sigma_\suf{\ell+1},\Diamonddot\varphi)\right\} \text{ for each $k \in \{1, \ldots, 4 \}$} .
\end{align}
\end{proposition}

It is important to highlight that Equation~\eqref{Eq:Exp2} does not only recur on $V_2$ but also on $V_1$ (an analogous observation is true for Equation~\eqref{Eq:Exp3}). 
In fact, by recurring on $V_1(\sigma_\suf{\ell},\boxdot\varphi)$ instead of $\sup_{k \geq \ell}{V_2(\sigma_\suf{k}, \varphi)}$, as one might have expected, we avoid the intermediate computation of $\sup_{k \geq \ell}{V_2(\sigma_\suf{k}, \varphi)}$ by the generalized Büchi automaton and, thereby, save auxiliary memory.
This is the key property that allows us to prevent an unduly growth in the size of the resulting Büchi automaton and to achieve the desired bound on the number of states.

\begin{proof}[Proof of Proposition~\ref{prop:expansion_rules}]
Equality~\eqref{Eq:Exp1} follows directly from the properties of $\inf$:
\begin{align*}
V_1(\sigma_\suf{\ell},\boxdot\varphi)=\inf_{i\ge \ell}V_1(\sigma_\suf{i},\varphi)&=\inf\left\{V_1(\sigma_\suf{\ell},\varphi),V_1(\sigma_\suf{\ell+1},\varphi),V_1(\sigma_\suf{\ell+2},\varphi),\hdots\right\}\\
&=\inf\left\{V_1(\sigma_\suf{\ell},\varphi),\inf\left\{V_1(\sigma_\suf{\ell+1},\varphi),V_1(\sigma_\suf{\ell+2},\varphi),\hdots\right\}\right\}\\
&=\min\left\{V_1(\sigma_\suf{\ell},\varphi),\inf_{i\ge \ell+1} V_1(\sigma_\suf{i},\varphi)\right\}\\
&=\min\left\{V_1(\sigma_\suf{\ell},\varphi),V_1(\sigma_\suf{\ell+1},\boxdot\varphi)\right\}. 
\end{align*}
A similar argument using the properties of $\sup$ shows that
\[V_2(\sigma_\suf{\ell},\boxdot\varphi)= \sup_{j\ge \ell}\inf_{i\ge j} V_2(\sigma_\suf{i},\varphi)=\max\left\{\inf_{i\ge \ell} V_2(\sigma_\suf{i},\varphi),\sup_{j\ge \ell+1}\inf_{i\ge j} V_2(\sigma_\suf{i},\varphi)\right\}. \]

To conclude the proof of Equality~\eqref{Eq:Exp2}, we need to replace the term $\inf_{i\ge \ell} V_2(\sigma_\suf{i},\varphi)$ inside the $\max$ by $\inf_{i\ge \ell} V_1(\sigma_\suf{i},\varphi)$; in other words, we must prove the last equality in the equation
\begin{equation}
	\begin{aligned}
		\max\left\{V_1(\sigma_\suf{\ell},\boxdot\varphi),V_2(\sigma_\suf{\ell+1},\boxdot\varphi)\right\} & = \max\left\{\inf_{i\ge \ell} V_1(\sigma_\suf{i},\varphi),\sup_{j\ge \ell+1}\inf_{i\ge j} V_2(\sigma_\suf{i},\varphi)\right\} \\
		& = \max\left\{\inf_{i\ge \ell} V_2(\sigma_\suf{i},\varphi),\sup_{j\ge \ell+1}\inf_{i\ge j} V_2(\sigma_\suf{i},\varphi)\right\} 
	\end{aligned}
	\label{Eq:Equality}
\end{equation}
holds for every sequence $\sigma\in \Sigma^\omega$, every rLTL$(\boxdot,\Diamonddot)$ formula $\varphi$, and any valuation $V$. 

To this end, we consider two separate cases. The first case is $\sup_{j\ge \ell+1}\inf_{i\ge j} V_2(\sigma_\suf{i},\varphi)=1$ and immediately leads to the desired equality:
\[ \max\left\{\inf_{i\ge \ell} V_1(\sigma_\suf{i},\varphi),\sup_{j\ge \ell+1}\inf_{i\ge j} V_2(\sigma_\suf{i},\varphi)\right\}=1=\max\left\{\inf_{i\ge \ell} V_2(\sigma_\suf{i},\varphi),\sup_{j\ge \ell+1}\inf_{i\ge j} V_2(\sigma_\suf{i},\varphi)\right\}. \]

The second case is $\sup_{j\ge \ell+1}\inf_{i\ge j} V_2(\sigma_\suf{i},\varphi)=0$ and the desired equality reduces to
\[ \inf_{i\ge \ell} V_1(\sigma_\suf{i},\varphi)=\inf_{i\ge \ell} V_2(\sigma_\suf{i},\varphi). \]
We now note that $\sup_{j\ge \ell+1}\inf_{i\ge j} V_2(\sigma_\suf{i},\varphi)=0$ implies $\inf_{i\ge \ell+1} V_2(\sigma_\suf{i},\varphi)=0$ which, in turn, implies $\inf_{i\ge \ell} V_2(\sigma_\suf{i},\varphi)=0$. Hence, to conclude the proof, we must show $\inf_{i\ge \ell} V_1(\sigma_\suf{i},\varphi)=0$. We recall that every element $b=(b_1,b_2,b_3,b_4)\in\B_4$ satisfies $b_1\le b_2$. In particular, we have $V_1(\sigma_\suf{i},\varphi)\le V_2(\sigma_\suf{i},\varphi)$ for every $i\in \N$, and it follows from the monotonicity properties of $\inf$ that
\[ \inf_{i\ge \ell} V_1(\sigma_\suf{i},\varphi)\le \inf_{i\ge \ell} V_2(\sigma_\suf{i},\varphi). \]
The proof of equality~\eqref{Eq:Exp2} is now finished by noting that the previous inequality and $\inf_{i\ge \ell} V_2(\sigma_\suf{i},\varphi)=0$ imply $\inf_{i\ge \ell} V_1(\sigma_\suf{i},\varphi)=0$.

The proof of Equality~\eqref{Eq:Exp3} is dual to the proof of Equality~\eqref{Eq:Exp2}, while the proof of Equality~\eqref{Eq:Exp4} is dual to the proof of Equality~\eqref{Eq:Exp1}.
\end{proof}

\subsection{rLTL$(\boxdot,\Diamonddot)$ and Büchi Automata}\label{sec:rLTL_to_Buechi}
It is well-known that one can construct for any LTL formula a (generalized) Büchi automaton that accepts exactly those infinite words satisfying the formula.
Our goal is to establish a similar connection between rLTL$(\boxdot,\Diamonddot)$ and generalized Büchi automata. As preparation, let us briefly recapitulate the definition of generalized Büchi automata and introduce basic notations.

\subsubsection{A Brief Recapitulation of Generalized Büchi Automata}
Intuitively, a generalized Büchi automaton is a (nondeterministic) Büchi automaton with a set of acceptance conditions (rather than just a single one). A formal definition is as follows.

\begin{definition}[Generalized Büchi automaton]
A \emph{generalized Büchi automaton} is a tuple $\mathcal A = (Q, \Sigma, q_0, \Delta, \mathcal F)$ consisting of a nonempty, finite set $Q$ of states, a (finite) input alphabet $\Sigma$, an initial state $q_0 \in Q$, a (nondeterministic) transition relation $\Delta \in Q \times \Sigma \times Q$, and a set $\mathcal F \subseteq 2^Q$ denoting the acceptance conditions.
\end{definition}

The \emph{run} of a generalized Büchi automaton on a word $\sigma \in \Sigma^\omega$ (also called \emph{input}) is an infinite sequence of states $\rho = q_0 q_1 \ldots \in Q^\omega$ satisfying $(q_i, \sigma(i), q_{i+1}) \in \Delta$ for all $i \in \mathbb N$ (note that each run starts in the initial state $q_0$).
Given a run $\rho = q_0 q_1 \ldots$, we denote the set of states occurring infinitely often during $\rho$ by $\Inf(\rho) = \{ q \in Q \mid \forall i \in \mathbb N~ \exists j \geq i \colon q_j = q\}$.
A run $\rho$ is called \emph{accepting} if $\Inf(\rho) \cap F \neq \emptyset$ for all $F \in \mathcal F$ (i.e., the run visits a state of each set $F \in \mathcal F$ infinitely often). The \emph{language} of a generalized Büchi automaton $\mathcal A$, denoted by $L(\mathcal A)$, is the set of all infinite words $\sigma \in \Sigma^\omega$ for which an accepting run of $\mathcal A$ exists.

\subsubsection{From rLTL$(\boxdot,\Diamonddot)$ to Generalized Büchi Automata}\label{sec:LTL_2_BA}
A classical translation of LTL formulas into generalized Büchi automata is based on the so-called $\varphi$-expansion: given an LTL formula $\varphi$, the $\varphi$-expansion of an infinite word $\sigma \in \Sigma^\omega$ tracks the evaluation of $\varphi$ and its subformulas at each position of $\sigma$. The key idea is to construct a generalized Büchi automaton that nondeterministically guesses the $\varphi$-expansion step-by-step when reading its input (and verifies the guess by means of its acceptance conditions). The automaton is constructed to accept an input $\sigma$ if and only if the $\varphi$-expansion signals that $W(\sigma, \varphi) = 1$.

Our approach follows a similar line and translates an rLTL$(\boxdot,\Diamonddot)$ formula $\varphi$ into a generalized Büchi automaton $\mathcal A_\varphi$.  However, since the value of an rLTL$(\boxdot,\Diamonddot)$ formula is not Boolean but an element of $\mathbb B_4$, we construct a generalized Büchi automaton without a dedicated initial state. Instead, we introduce for each $b \in \mathbb B_4$ a state $q_b$ and construct $\mathcal A_\varphi$ such that it accepts an input $\sigma$ starting in state $q_b$ if and only if $V(\sigma, \varphi) = b$. In this way, we can easily determine the value of an arbitrary word by simply checking from which of the states $q_b$ it is accepted (it is, by construction, accepted from exactly one of these states).

As the classical translation, our translation is based on the notion of $\varphi$-expansion, which records the value of each subformula of $\varphi$ on the given word. The set of sub-formulas of an rLTL$(\boxdot,\Diamonddot)$ formula, called closure, is defined next.

\begin{definition}[Closure]
Let $p \in \mathcal P$ an atomic proposition and $\varphi, \psi$ two rLTL$(\boxdot,\Diamonddot)$ formulas. The \emph{closure} of an rLTL$(\boxdot,\Diamonddot)$ formula, denoted by $\closure$, is inductively defined as follows:
\begin{itemize}
	\item $\closure(p) = \{ p \}$;
	\item $\closure(\lnot \varphi) = \{\lnot \varphi\} \cup \closure(\varphi)$;
	\item $\closure(\varphi \land \psi) = \{\varphi \land \psi \} \cup \closure(\varphi) \cup \closure(\psi)$;
	\item $\closure(\varphi \lor \psi) = \{\varphi \lor \psi \} \cup \closure(\varphi) \cup \closure(\psi)$;
	\item $\closure(\varphi \Rightarrow \psi) = \{\varphi \Rightarrow \psi \} \cup \closure(\varphi) \cup \closure(\psi)$;
	\item $\closure(\Diamonddot \varphi) = \{\Diamonddot \varphi \} \cup \closure(\varphi)$; and
	\item $\closure(\boxdot \varphi) = \{ \boxdot \varphi \} \cup \closure(\varphi)$.
\end{itemize}
\end{definition}

Having introduced the closure of an rLTL$(\boxdot,\Diamonddot)$ formula $\varphi$, we can now define the $\varphi$-expansion.

\begin{definition}[$\varphi$-expansion]\label{def:phi_expansion}
Let $\varphi$ be an rLTL$(\boxdot,\Diamonddot)$ formula. The \emph{$\varphi$-expansion} of an infinite word $\sigma \in \Sigma^\omega$ is a mapping $\eta \colon\closure(\varphi) \times \mathbb N \to \mathbb B_4$ satisfying $\eta(\psi, i) = V(\sigma_\suf{i}, \psi)$ for all $\psi \in \closure(\varphi)$ and $i \in \mathbb N$. 
\end{definition}

Note that the $\varphi$-expansion is unique for a given word and subsumes the valuation of $\varphi$ in the sense that $V(\varphi, \sigma) = \eta(\varphi, 0)$. Although the definition of the $\varphi$-expansion is not constructive, we can introduce constraints that completely characterize the $\varphi$-expansion of a given word.
The pivotal idea is to impose constraints for local consistency (e.g., $\eta(\lnot \psi, i)$ for $\psi \in \cl(\varphi)$ at some position $i \in \mathbb N$ has to be $\overline{\eta(\psi, i)}$) and to exploit the expansion rules of Proposition~\ref{prop:expansion_rules} to relate $\eta(\psi, i)$ and $\eta(\psi, i+1)$. As in the case of valuations $V$, we use the shorthand-notation $\eta_j(\psi, i)$ instead of the more verbose expression $\pi_j(\eta(\psi, i))$.

In the following, let $\psi \in \closure(\varphi)$ and $i \in \mathbb N$. The first type of constraints (\emph{local constraints}) are as follows:

\begin{enumerate}[label=A\arabic*), ref=A\arabic*]
	\item \label{cond:local1} If $\psi = p$, then $\eta(\psi, i) = \begin{cases} 0000 & \text{if $p \notin \sigma(i)$; and} \\ 1111 & \text{if $p \in \sigma(i)$.} \end{cases}$
	\item \label{cond:local2} If $\psi = \neg \psi_1$, then $\eta(\psi, i) = \overline{\eta(\psi_1, i)}$.
	\item \label{cond:local3} If $\psi = \psi_1 \wedge \psi_2$, then $\eta(\psi, i) = \min{\{\eta(\psi_1, i), \eta(\psi_2, i)\}}$.
	\item \label{cond:local4} If $\psi = \psi_1 \vee \psi_2$, then $\eta(\psi, i) = \max{\{\eta(\psi_1, i), \eta(\psi_2, i)\}}$.
	\item \label{cond:local5} If $\psi = \psi_1 \Rightarrow \psi_2$, then $\eta(\psi, i) = \eta(\psi_1, i) \rightarrow \eta(\psi_2, i)$.
	%
	%
	\item \label{cond:local7} If $\psi = \Diamonddot \psi_1$, then $\eta(\psi, i) = (b_1, b_2, b_3, b_4)$ where $b_j = \max{ \bigl \{ \eta_j(\psi_1, i), \eta_j(\psi, i+1) \bigr \} }$ for $j \in \{1, \ldots, 4\}$.
	\item \label{cond:local8} If $\psi = \boxdot \psi_1$, then $\eta(\psi, i) = (b_1, b_2, b_3, b4)$ where
	\begin{enumerate}
		\item \label{cond:local8a} $b_1 = \min{ \bigl \{ \eta_1(\psi_1, i), \eta_1(\psi, i+1) \bigr \} }$;
		\item \label{cond:local8b} $b_2 = \max{ \bigl \{ b_1, \eta_2(\psi, i+1) \bigr \} }$;
		\item \label{cond:local8c}  $b_3 = \min{ \bigl \{ b_4, \eta_3(\psi, i+1) \bigr \} }$; and
		\item \label{cond:local8d} $b_4 = \max{ \bigl \{ \eta_4(\psi_1, i), \eta_4(\psi, i+1) \bigr \} }$.
	\end{enumerate}
\end{enumerate}
To ensure satisfaction of the subformulas involving the temporal operators $\Diamonddot$ and $\boxdot$, we add the following further constraints (\emph{non-local constraints}). These constraints are derived from the expansion rules, and we later translate them into Büchi conditions.
\begin{enumerate}[label=B\arabic*), ref=B\arabic*]
	\item \label{cond:nonlocal1} For each $\Diamonddot \psi \in \closure(\varphi)$ and $j \in \{1, \ldots, 4\}$, there exists no $k \in \mathbb N$ such that for every $\ell \geq k$ both $\eta_j(\Diamonddot \psi, \ell) = 1$ and $\eta_j(\psi, \ell) = 0$.
	\item \label{cond:nonlocal2} For each $\boxdot \psi \in \closure(\varphi)$,
	\begin{enumerate}
		\item \label{cond:nonlocal2a} there exists no $k \in \mathbb N$ such that for every $\ell \geq k$ both $\eta_1(\boxdot \psi, \ell) = 0$ and $\eta_1(\psi, \ell) = 1$;
		\item \label{cond:nonlocal2b} there exists no $k \in \mathbb N$ such that for every $\ell \geq k$ both $\eta_2(\boxdot \psi, \ell) = 1$ and $\eta_1(\boxdot \psi, \ell) = 0$;
		\item \label{cond:nonlocal2c} there exists no $k \in \mathbb N$ such that for every $\ell \geq k$ both $\eta_3(\boxdot \psi, \ell) = 0$ and $\eta_4(\boxdot \psi, \ell) = 1$; and
		\item \label{cond:nonlocal2d} there exists no $k \in \mathbb N$ such that for every $\ell \geq k$ both $\eta_4(\boxdot \psi, \ell) = 1$ and $\eta_4(\psi, \ell) = 0$.
	\end{enumerate}
\end{enumerate}

Let us now show that these constraints indeed completely characterize the $\varphi$-expansion of a given word. 

\begin{lemma}\label{lem:compatibility}
Given an rLTL$(\boxdot,\Diamonddot)$ formula $\varphi$ over the atomic propositions $\mathcal P$ and an infinite word $\sigma \in \Sigma^\omega$ where $\Sigma = 2^\mathcal P$, let $\eta \colon \closure(\varphi) \times \mathbb N \to \mathbb B_4$ be a mapping that satisfies the compatibility constraints \ref{cond:local1} to \ref{cond:nonlocal2}. Then, $\eta$ is uniquely determined, and it is, in fact, the $\varphi$-expansion of $\sigma$.
\end{lemma}

\begin{proof}
To prove Lemma~\ref{lem:compatibility}, we need to establish that $V(\sigma_\suf{i}, \psi) = \eta(\psi, i)$ holds for all $\psi \in \closure(\varphi)$ and $i \in \mathbb N$. The proof proceeds by structural induction over the subformulas of $\varphi$.

\begin{description}[font={\normalfont\itshape}]
	\item[Base case] In the case of atomic propositions, the claim holds by definition of $V$.
	\item[Induction step] In the case of the operators $\lnot$, $\lor$, $\land$, and $\Rightarrow$, the claim follows immediately from applying the induction hypothesis and by definition of $V$.
	
	In the case of $\psi = \Diamonddot \psi_1$, a straightforward induction that applies
	\begin{itemize}
		\item Condition~\ref{cond:local7};
		\item the expansion rule for $\Diamonddot$ (see Proposition~\ref{prop:expansion_rules}, Equation~\eqref{Eq:Exp4}); and
		\item the induction hypothesis for $\psi_1$ (i.e., $V(\sigma_\suf{i}, \psi_1) = \eta(\psi_1, i)$ for all $i \in \mathbb N$)
	\end{itemize}
	shows that the following is true for each $j \in \{1, \ldots, 4\}$: if $\eta_j(\psi_1, k) = 1$ for a $k \in \mathbb N$, then  $\eta_j(\psi, \ell) = 1$ and, hence, $V_j(\sigma_\suf{\ell}, \psi) = \eta_j(\psi, \ell)$ for all $\ell \leq k$. Therefore, if infinitely many $k$ with $\eta_j(\psi_1, k) = 1$ exist, then $V_j(\sigma_\suf{i}, \psi) = \eta_j(\psi, i)$ for all $i \in \mathbb N$. If this is not the case, then there exists a $k \in \mathbb N$ such that $\eta_j(\psi_1, \ell) = 0$ for all $\ell \geq k$. Then, Condition~\ref{cond:nonlocal1} asserts for all $\ell \geq k$ that $\eta_j(\psi, \ell) = 0$ and, hence, $V_j(\sigma_\suf{\ell}, \psi) = \eta_j(\psi, \ell)$ is satisfied by the semantics of $\Diamonddot$ and the induction hypothesis for $\psi_1$; this, in turn, implies $V_j(\sigma_\suf{i}, \psi) = \eta_j(\psi, i)$ for all $i \in \mathbb N$. These arguments are true for all $j \in \{1, \ldots, 4 \}$ and, therefore, $V(\sigma_\suf{i}, \psi) = \eta(\psi, i)$ holds for all $i \in \mathbb N$.

	The case $\psi = \boxdot \psi_1$ can be proven using similar arguments as in the case of the $\Diamonddot$-operator, but the semantics of $\boxdot$ requires to split the proof into four parts and prove $V_j(\sigma_\suf{i}, \psi) = \eta_j(\psi, i)$ individually for each $j \in \{1, \ldots, 4\}$. So as not to clutter this proof too much, we provide a detailed proof for $j = 1$ and skip the remaining. However, it is important to note that the claim needs to be proven first for $j=1$ and $j=4$ since the proofs for $j=2$ and $j=3$ rely thereon (the expansion rules recur on $V_1(\sigma_\suf{i}, \psi)$ and $V_4(\sigma_\suf{i}, \psi)$, respectively).
  
	To prove $V_1(\sigma_\suf{i}, \psi) = \eta_1(\psi, i)$ for all $i \in \mathbb N$, we first observe that $\eta_1(\psi_1, k)= 0$ for a $k \in \mathbb N$ implies $V_1(\sigma_\suf{\ell}, \psi) = \eta_1(\psi, \ell)$ for all $\ell \leq k$; analogous to the case of the operator $\Diamonddot$, an induction using Condition~\ref{cond:local8a}, the expansion rule for $\boxdot$ (see Proposition~\ref{prop:expansion_rules}, Formula~\eqref{Eq:Exp1}), and the induction hypothesis for $\psi_1$ establishes this. Therefore, if infinitely many $k$ with $\eta_1(\psi_1, k) = 0$ exist, then $V_1(\sigma_\suf{i}, \psi) = \eta_1(\psi, i)$ for all $i \in \mathbb N$. If this is not the case, then there exists a $k \in \mathbb N$ such that $\eta_1(\psi_1, \ell) = 1$ for all $\ell \geq k$. Then, Condition~\ref{cond:nonlocal2a} asserts for all $\ell \geq k$ that $\eta_1(\psi, \ell)) = 1$ and, hence, $V_1(\sigma_\suf{\ell}, \psi) = \eta_1(\psi, \ell)$ is satisfied by the semantics of $\boxdot$ and the induction hypothesis of $\psi_1$. This implies $V_1(\sigma_\suf{i}, \psi) = \eta_1(\psi, i)$ for all $i \in \mathbb N$.

	As mentioned above, the case $j = 4$ and the subsequent cases $j=2$ and $j=3$ are analogous. \qedhere
\end{description}
\end{proof}

We are now ready to define a generalized Büchi automaton $\mathcal A_\varphi$. The states of $\mathcal A_\varphi$ are mappings $\mu \colon \closure(\varphi) \to \mathbb B_4$, which encode the $\varphi$-expansion of $\sigma$ in the sense that the sequence of states $\mu_0, \mu_1, \ldots$ constituting an accepting run on $\sigma$ satisfies $\mu_i(\psi) = \eta_i(\psi)$ for all $i \in \mathbb N$ and $\psi \in \closure(\varphi)$. Clearly, the only states (i.e., mappings $\mu$) of interest are those consistent with the local compatibility constraints~\ref{cond:local1} to \ref{cond:local5}.\kern-.06em\footnote{By this we mean that the conditions are satisfied if we substitute $\mu$ for $\eta$.} Thus, in order to ease the following definition, we denote the set of such mappings by $S$. Note that the cardinality of $S$ is bounded by $|\mathbb B_4|^{|\closure(\varphi)|} = 5^{|\closure(\varphi)|}$.

When reading an input-word, the automaton $\mathcal A_\varphi$ uses its transitions to verify that its guess satisfies the local constraints and uses its acceptance condition to verify the non-local constraints. The latter is achieved by adding a Büchi condition for each of the Conditions~\ref{cond:nonlocal1} to \ref{cond:nonlocal2d}, which translate the respective condition in a straightforward manner. Hence, the number of acceptance conditions is exactly four times the number of subformulas of type $\Diamonddot$ and $\boxdot$.

Finally, it is important to note that we define the automaton without an initial state. Instead, we introduce a state $q_b$ for each $b \in \mathbb B_4$ with the property that $\mathcal A_\varphi$ accepts a word $\sigma \in \Sigma^\omega$ when starting in the state $q_b$ if and only if $V(\sigma, \varphi) = b$. In other words, an accepting run starting in $q_b$ signals that $\varphi$ evaluates on $\sigma$ to $b$.

\begin{definition}[Automaton $\mathcal A_\varphi$]\label{def:buechi_automaton}
Let $\varphi$ be an rLTL$(\boxdot,\Diamonddot)$ formula over the atomic propositions $\mathcal P$. Additionally, let $\Sigma = 2^\mathcal P$, $a \in \Sigma$, and $S$ be the set of functions $\mu \colon \closure(\varphi) \to \mathbb B_4$ that satisfy Conditions~\ref{cond:local1} to \ref{cond:local5}. We define the \emph{generalized Büchi automaton $\mathcal A_\varphi = (Q, \Sigma, \Delta, \mathcal F)$} as follows:

\begin{itemize} 
	\item $Q = \{ q_b \mid b \in \mathbb B_4\} \cup S$;
	\item the transition relation is defined by:
	\begin{itemize}
		\item $(q_b, a, \mu) \in \Delta$ if and only if $\mu(\varphi) = b$ and $\mu(p) = \begin{cases} 1111 & \text{if $p \in a \cap \cl(\varphi)$; and} \\ 0000 & \text{if $p \in \cl(\varphi) \setminus a$;} \end{cases}$
		\item $(\mu, a, \mu') \in \Delta$ if and only if the pair $(\mu$, $\mu')$ satisfies Conditions~\ref{cond:local7} and \ref{cond:local8} as well as\\$\mu'(p) = \begin{cases} 1111 & \text{if $p \in a \cap \cl(\varphi)$; and} \\ 0000 & \text{if $p \in \cl(\varphi) \setminus a$;} \end{cases}$
	\end{itemize}
	\item $\mathcal F$ is the union of the following sets:
	\begin{itemize}
		\item for each $\Diamonddot \psi \in \closure(\varphi)$, we introduce for each $j \in \{1, \ldots, 4\}$ the set
		\[ F_{\Diamonddot \psi, j} = \{ \mu \in S \mid \pi_j(\mu(\Diamonddot \psi)) = 0 \text{ or } \pi_j(\mu(\psi)) = 1 \}; \]
		\item for each $\boxdot \psi \in \closure(\varphi)$, we introduce the sets
		\begin{align*}
			F_{\boxdot \psi, 1} = {} & \{ \mu \in S \mid \pi_1(\mu(\boxdot \psi)) = 1 \text{ or } \pi_1(\mu(\psi)) = 0 \}; \\
			F_{\boxdot \psi, 2} = {} & \{ \mu \in S \mid \pi_2(\mu(\boxdot \psi)) = 0 \text{ or } \pi_1(\mu(\boxdot \psi)) = 1 \}; \\
			F_{\boxdot \psi, 3} = {} & \{ \mu \in S \mid \pi_3(\mu(\boxdot \psi)) = 1 \text{ or } \pi_4(\mu(\boxdot \psi)) = 0 \}; \text{and} \\
			F_{\boxdot \psi, 4} = {} & \{ \mu \in S \mid \pi_4(\mu(\boxdot \psi)) = 0 \text{ or } \pi_4(\mu(\psi)) = 1 \}.
		\end{align*}
	\end{itemize}
\end{itemize}
\end{definition}

Definition~\ref{def:buechi_automaton} ensures that $\mathcal A_\varphi$ accepts $\sigma \in \Sigma^\omega$ if and only if there exists a run $q_b, \mu_0, \mu_1, \ldots$ that visits each $F \in \mathcal F$ infinitely often. As an example, suppose that a run visits the set $F_{\Diamonddot \psi, 1}$ for $\Diamonddot \psi \in \closure(\varphi)$ infinitely often (i.e., $\pi_1(\mu_i(\Diamonddot \psi)) = 0$ or $\pi_1(\mu_i(\psi)) = 1$ holds for infinitely many  $i \in \mathbb N$). This means that it never happens that from some $k \in \mathbb N$ onward both $\pi_1(\mu_k(\Diamonddot \psi)) = 1$ and $\pi_1(\mu_k(\psi)) = 0$. Hence, Condition~\ref{cond:nonlocal1} is fulfilled. Similarly, the remaining sets $F \in \mathcal F$ make sure that Conditions~\ref{cond:nonlocal1} and \ref{cond:nonlocal2} are indeed satisfied. Moreover, the definition of $\Delta$ ensures that Conditions~\ref{cond:local1} to \ref{cond:local8d} are satisfied along an accepting run of $\mathcal A_\varphi$ on $ \sigma$ and, therefore, this run in fact forms the $\varphi$-expansion of $\sigma$ (and is unique). Finally, by using different initial states, we make sure that $\mathcal A_\varphi$ accepts $\sigma$ starting from $q_b$ if only if $b = V(\sigma, \varphi)$ (since all outgoing transitions lead to states $\mu$ with $\mu(\varphi) = b$). As a consequence, we obtain the following theorem.

\begin{theorem}\label{thm:buechi_automaton_correct}
Let $\varphi$ be an rLTL$(\boxdot,\Diamonddot)$ formula over the set $\mathcal P$ of atomic propositions, $\Sigma = 2^\mathcal P$, and $b \in \mathbb B_4$. Then, $\mathcal A_\varphi$ accepts $\sigma \in \Sigma^\omega$ when starting in state $q_b$ if and only if $V(\sigma, \varphi) = b$.
\end{theorem}

For notational convenience, we denote the generalized Büchi automaton $\mathcal A_\varphi$ with initial state $q_b$ by $\mathcal A_\varphi^b$. We finish the discussion with a remark about the size of the automaton $\mathcal A_\varphi$.

\begin{remark}\label{rem:size_buechi_automaton}
The automaton $\mathcal A_\varphi$ has $5^{|\cl(\varphi)|} + 4$ states and at most $4 \cdot |\cl(\varphi)|$ acceptance sets.
\end{remark}

\subsection{Model Checking}
\label{sec:model_checking}
Broadly speaking, the model checking problem asks whether the model of a given system exhibits a specified behavior (which is described as an rLTL$(\boxdot,\Diamonddot)$ formula in our case). Usually, a system is modeled as a Kripke structure, which is, for the sake of model checking, translated into a Büchi automaton whose language corresponds to the unraveling of the Kripke structure. For reasons of simplicity, we consider a system---more precisely, model thereof---to be given directly as a (generalized) Büchi automaton. This leads to the following formulation of the model checking problem.

\begin{problem}[Model checking]\label{prob:model_checking_exact}
Let $\varphi$ be an rLTL$(\boxdot,\Diamonddot)$ formula over the set $\mathcal P$ of atomic propositions, let $\mathcal A$ be a generalized Büchi automaton over the alphabet $2^\mathcal P$, and let $b \in \mathbb B_4$. Does $V(\sigma, \varphi) = b$ hold for all $\sigma \in L(\mathcal A)$?
\end{problem}

Our translation of rLTL$(\boxdot,\Diamonddot)$ formulas into a generalized Büchi automaton provides a straightforward means to answer the model checking problem: one simply constructs $\mathcal A_\varphi$ and checks $L(\mathcal A) \subseteq L(\mathcal A_\varphi^b)$. 
However, the naive attempt to check this inclusion (i.e., checking whether $L(\mathcal A) \cap ( \Sigma^\omega \setminus L(\mathcal A_\varphi^b)) = \emptyset$ holds) would require to complement $\mathcal A_\varphi^b$, which we clearly want to avoid due to the inevitable exponential blowup; moreover, note that the equality $\Sigma^\omega \setminus L(\mathcal A_\varphi^b) = L(\mathcal A_{\neg \varphi}^b)$ does not hold in general. Instead, we exploit the property that one obtains a generalized Büchi automaton accepting exactly the words with value $b’ \in \mathbb B_4$ from $\mathcal A_\varphi$ by designating $q_{b’}$ as the initial state. This fact allows us to write the complement of $L(\mathcal A_\varphi^b)$ as the union
\[ \Sigma^\omega \setminus L(\mathcal A_\varphi^b) = \bigcup_{b’ \in \mathbb B_4 \setminus \{ b \}} L(\mathcal A_\varphi^{b’}). \]

In addition, we can easily modify $\mathcal A_\varphi$ to accept this union:
\begin{enumerate}
	\item we add a new state, say $q_0$, and designate it as the initial state; and
	\item we add the $\varepsilon$-transitions $(q_0, \varepsilon, q_{b’})$ for all $b’ \in \mathbb B_4 \setminus \{ b \}$, which can subsequently be removed in the same manner as for finite automata with $\varepsilon$-transitions (see, e.g., Hopcroft and Ullman~\cite{Hopcroft}).
\end{enumerate}

In summary, we obtain the following result.

\begin{theorem}\label{thm:model_checking_exact}
One can decide the model checking problem (Problem~\ref{prob:model_checking_exact}) for $\mathcal A = (Q, \Sigma, q_0, \Delta, \mathcal F)$ and $\varphi$ in time
\[ \mathcal O \bigl( (|\mathcal F| + |\cl(\varphi)|) \cdot |Q| \cdot 5^{|\cl(\varphi)|} \bigr). \]
\end{theorem}

\begin{proof}[Proof of Theorem~\ref{thm:model_checking_exact}]
Let $\varphi$ be an rLTL$(\boxdot,\Diamonddot)$ formula over the atomic propositions $\mathcal P$, $\mathcal A = (Q, \Sigma, q_0, \Delta, \mathcal F)$ a generalized Büchi automaton over the alphabet $2^\mathcal P$, and $b \in \mathbb B_4$.

First, it is not hard to verify that
\begin{align*}
V(\sigma, \varphi) = b \text{ for all $\sigma \in L(\mathcal A)$} & ~\text{ if and only if }~ L(\mathcal A) \subseteq L(\mathcal A_\varphi^b) \\
& ~\text{ if and only if }~ L(\mathcal A) \cap \bigl( \Sigma^\omega \setminus L(\mathcal A_\varphi^b) \bigr) = \emptyset \\
& ~\text{ if and only if }~ L(\mathcal A) \cap \bigcup_{b’ \in \mathbb B_4 \setminus \{ b \}} L(\mathcal A_\varphi^{b’}) = \emptyset.
\end{align*}

Moreover, it follows from Theorem~\ref{thm:buechi_automaton_correct} that the construction sketched above in fact results in a generalized Büchi automaton $\mathcal B$ accepting $\bigcup_{b’ \in \mathbb B_4 \setminus \{ b \}} L(\mathcal A_\varphi^{b’})$. Since $\mathcal A_\varphi$ has $5^{|\cl(\varphi)|} + 4$ states and at most $4 \cdot |\cl(\varphi)|$ acceptance sets, the automaton $\mathcal B$ has $5^{|\cl(\varphi)|} + 5$ states and also at most $4 \cdot |\cl(\varphi)|$ acceptance sets.

Second, given two generalized Büchi automata $\mathcal A_1 = (Q_1, \Sigma, q_0^1, \Delta_1, \mathcal F_1)$ and $\mathcal A_2 = (Q_2, \Sigma, q_0^2, \Delta_2, \mathcal F_2)$, it is well-known that one can construct a generalized Büchi automaton accepting $L(\mathcal A_1) \cap L(\mathcal A_2)$ using a simple product construction (see, e.g., Perrin and Pin~\cite{InfiniteWordsBook}). This construction results in an automaton with $|Q_1| \cdot |Q_2|$ states and $|\mathcal F_1| + |\mathcal F_2|$ acceptance sets. Since $\mathcal B$ consists of $5^{|\cl(\varphi)|} + 5$ states and has at most $4 \cdot |\cl(\varphi)|$ acceptance sets, this implies that one can construct a generalized Büchi automaton $\mathcal C$ with $L(\mathcal C) = L(\mathcal A) \cap L(\mathcal B)$ consisting of $|Q| \cdot (5^{|\cl{\varphi}|} + 5)$ states and at most $|\mathcal F| + 4 \cdot |\cl(\varphi)|$ acceptance sets.

Finally, it is left to check whether $L(\mathcal C) = \emptyset$. This problem is fundamental in LTL model checking, and there exist efficient algorithms that solve this problem in time linear in the product of the number of states of the input automaton and the number of its acceptance sets (see, e.g., Baier and Katoen~\cite{Baier:2008:PMC:1373322}). Hence, one can solve Problem~\ref{prob:model_checking_exact} in $\mathcal O \bigl( (|\mathcal F| + |\cl(\varphi)|) \cdot |Q| \cdot 5^{|\cl(\varphi)|} \bigr)$ time.
\end{proof}

If the answer to Problem~\ref{prob:model_checking_exact} is negative, it is natural to ask a weaker question, namely whether every word accepted by the Büchi automaton in question has at least value $b$.

\begin{problem}[At-least model checking]\label{prob:model_checking_at_least}
Let $\varphi$ be an rLTL$(\boxdot,\Diamonddot)$ formula over the set $\mathcal P$ of atomic propositions, $\mathcal A$ a generalized Büchi automaton over the alphabet $2^\mathcal P$, and $b \in \mathbb B_4$. Does $V(\sigma, \varphi) \geq b$ hold for all $\sigma \in L(\mathcal A)$?
\end{problem}

Using the same ideas as above, one can reduce deciding  the at-least model checking problem to checking the inclusion $\mathcal L(\mathcal A) \subseteq \bigcup_{b’ \in \mathbb B_4, b’ \geq b} L(\mathcal A_\varphi^{b’})$. Again, we avoid the complement by checking $L(\mathcal A) \cap \bigcup_{b’ \in \mathbb B_4, b’ < b} L(\mathcal A_\varphi^{b’}) = \emptyset$ 
instead, which immediately yields the next result.

\begin{corollary}\label{cor:model_checking_at_least}
One can decide the at-least model checking problem (Problem~\ref{prob:model_checking_at_least}) for $\mathcal A = (Q, \Sigma, q_0, \Delta, \mathcal F)$ and $\varphi$ in time $\mathcal O \bigl( (|\mathcal F| + |\cl(\varphi)|) \cdot |Q| \cdot 5^{|\cl(\varphi)|} \bigr)$.
\end{corollary}

The many valued semantics of rLTL$(\boxdot,\Diamonddot)$ allows posing optimization problems as well; for instance, a user might be interested in the largest value that a system guarantees.
Repeatedly solving the at-least model checking problem for decreasing values of $b$ already solves this problem, which is summarized in the following remark.

\begin{remark}
Given an rLTL$(\boxdot,\Diamonddot)$ formula $\varphi$ over the set $\mathcal P$ of atomic propositions and a generalized Büchi automaton $\mathcal A = (Q, \Sigma, q_0, \Delta, \mathcal F)$ over the alphabet $2^\mathcal P$, one can compute the largest $b \in \mathbb B_4$ such that $V(\sigma, \varphi) \geq b$ for all $\sigma \in L(\mathcal A)$ in time $\mathcal O \bigl( (|\mathcal F| + |\cl(\varphi)|) \cdot |Q| \cdot 5^{|\cl(\varphi)|} \bigr)$.
\end{remark}

\subsection{Reactive Synthesis}
\label{sec:synthesis}
In the context of reactive synthesis, we consider infinite-duration two-player games over finite graphs with rLTL$(\boxdot,\Diamonddot)$ winning conditions. In particular, we show, given a game with rLTL$(\boxdot,\Diamonddot)$ winning condition, how to construct a finite-state winning strategy.
Throughout this section, we assume familiarity with games over finite graphs and follow the definitions and notations of Grädel, Thomas, and Wilke~\cite{DBLP:conf/dagstuhl/2001automata}.

We consider games of the following kind.
\begin{definition}[rLTL$(\boxdot,\Diamonddot)$ games]
Let $\mathcal P$ be a finite set of atomic propositions. An \emph{rLTL$(\boxdot,\Diamonddot)$ game} is a pair $\mathfrak G = (\mathcal G, (\varphi, B))$ consisting of 
\begin{itemize}
	\item a finite, labeled \emph{game graph} $\mathcal G = (V, E, \lambda)$ where $V$ is a finite set of vertices that is partitioned into two disjoint sets $V_0, V_1 \subseteq V$, $E \subseteq V \times V$ is an edge relation, and $\lambda \colon V \to 2^\mathcal P$ is a function labeling each vertex with atomic propositions; and
	\item a pair $(\varphi, B)$ consisting of an rLTL$(\boxdot,\Diamonddot)$ formula $\varphi$ over $\mathcal P$ and a set $B \subseteq \mathbb B_4$ (this pair constitutes the \emph{winning condition} as we formalize shortly). 
\end{itemize}
\end{definition}

An rLTL$(\boxdot,\Diamonddot)$ game is played as usual by two players, Player~0 and Player~1, who construct a \emph{play} $\rho = v_0 v_1 \ldots \in V^\omega$ (i.e., an infinite sequence of vertices) by moving a token along the edges of the game graph. A play $\rho = v_0 v_1 \ldots$ induces an infinite word $\lambda(\rho) = \lambda(v_0) \lambda(v_1) \ldots \in {(2^\mathcal P)}^\omega$, and the value of the formula $\varphi$ on $\lambda(\rho)$ is used to determine the winner of the play. More precisely, we call a play $\rho \in V^\omega$ \emph{winning for Player~0} if $V(\lambda(\rho), \varphi) \in B$; symmetrically, we call a play \emph{winning for Player~1} if it is \emph{not winning} for Player~0.

A strategy of Player~$i$, $i \in \{ 0, 1 \}$, is a mapping $f \colon V^\ast V_i \to V$  that prescribes  the next move of Player~$i$ depending on the finite play played thus far. We call a strategy $f$ of Player~$i$ \emph{winning} from a state $v_0 \in V$ if all plays that start in $v_0$ and that are played according to $f$ are winning for Player~$i$. Moreover, we call a (winning) strategy a \emph{finite-state strategy} if there exists a finite-state machine computing it in the usual sense (see Grädel, Thomas, and Wilke~\cite{DBLP:conf/dagstuhl/2001automata} for further details).
Computing a finite-state winning strategy for Player~0 is the objective of the remainder of this section.

It is not hard to verify that determinacy of rLTL$(\boxdot,\Diamonddot)$ and the existence of a finite-state winning strategy follows from Theorem~\ref{thm:buechi_automaton_correct} and the determinacy of Büchi games, which leads to the following remark.

\begin{remark}
rLTL$(\boxdot,\Diamonddot)$ games are determined with finite-state winning strategies.
\end{remark}

Given an rLTL$(\boxdot,\Diamonddot)$ game an a vertex $v \in V$, we are interested in \emph{solving the game} (i.e., in deciding which player has a winning strategy from $v$ and in computing such a strategy), which is formalized next.

\begin{problem}[Determining the winner]\label{prob:synthesis_decision}
Let an rLTL$(\boxdot,\Diamonddot)$ game $\mathfrak G = (\mathcal G, (\varphi, B))$ over the set $V$ of vertices and a vertex $v_0 \in V$ be given. Determine the player who has a winning strategy from vertex $v_0$.
\end{problem}

\begin{problem}[Strategy synthesis]\label{prob:synthesis_strategy}
Let an rLTL$(\boxdot,\Diamonddot)$ game $\mathfrak G = (\mathcal G, (\varphi, B))$ over the set $V$ of vertices and a vertex $v_0 \in V$ be given. Compute a winning strategy from vertex $v_0$.
\end{problem}

To solve these problems, we follow the Safra-based approach using the following four-step process:
\begin{enumerate}
	\item We construct a (nondeterministic) Büchi automaton $\mathcal B_\varphi^B$ with $L(\mathcal B_\varphi^B) = \{ \sigma \in (2^\mathcal P)^\omega \mid V(\sigma, \varphi) \in B\}$.
	\item We determinize $\mathcal B_\varphi^B$ using Safra's construction~\cite{DBLP:conf/focs/Safra88}, resulting in a (deterministic) Rabin automaton\footnote{A Rabin automaton is a tuple $\mathcal C = (Q, \Sigma, q_0, \delta, \Omega)$ where $Q$, $\Sigma$, and $q_0$ are as in Büchi automata, $\delta \colon Q \times \Sigma \to Q$ is a (deterministic) transition function, and $\Omega \subseteq 2^Q \times 2^Q$ is the acceptance condition. The \emph{run} of a Rabin automaton on a word $ \sigma \in \Sigma^\omega$ is an infinite sequence of states $\rho = q_0 q_1 \ldots$ satisfying $\delta(q_i, \sigma(i)) = q_{i+1}$ for all $i \in \mathbb N$. A run $\rho$ is called \emph{accepting} if there exists a pair $(E, F) \in \Omega$ such that $E \cap \Inf(\rho) = \emptyset$ and $F \cap \Inf(\rho) \neq \emptyset$.} $\mathcal C_\varphi^B$ that is language-equivalent to $\mathcal B_\varphi^B$.
	\item We construct a Rabin game\footnote{A Rabin game is a game played over an unlabeled game graph $\mathcal G = (V, E)$ with nonempty, finite set $V$ of vertices and directed edge relation $E \subseteq V \times V$. The winning condition of a Rabin game is a set $\Omega \subseteq 2^V \times 2^V$, and a play $\rho = v_0 v_1 \ldots \in V^\omega$ is said to be winning for Player~0 if there exists a pair $(E, F) \in \Omega$ such that $E \cap \Inf(\rho) = \emptyset$ and $F \cap \Inf(\rho) \neq \emptyset$; by slight abuse of notation, $\Inf(\rho)$ here corresponds to the set of all vertices occurring infinitely often in the play $\rho$.} $\mathfrak G'$ by taking the product of the game graph $\mathcal G$ and the Rabin automaton $\mathcal C_\varphi^B$.
	\item We apply standard techniques to solve $\mathfrak G'$, which allows us to decide which player has a winning strategy from $v$ and to construct a winning strategy for the corresponding player.
\end{enumerate}

Let us now sketch these steps.

\paragraph{\emph{Step 1}}
The construction of Section~\ref{sec:LTL_2_BA} can easily be adapted to produce a (nondeterministic) generalized Büchi automaton $\mathcal A_\varphi^B$ with $L(\mathcal A_\varphi^B) = \{ \sigma \in (2^\mathcal P)^\omega \mid V(\sigma, \varphi) \in B\}$; this automaton comprises $5^{|\cl(\varphi)|} + 5$ states and at most $4 \cdot |\cl(\varphi)|$ acceptance sets. Subsequently, we construct a nondeterministic Büchi automaton $\mathcal B_\varphi^B$ accepting the same language; the standard conversion results in a Büchi automaton that comprises $\mathcal O(4 \cdot |\cl(\varphi)| \cdot (5^{|\cl(\varphi)|} + 5))$ states.

\paragraph{\emph{Step 2}}
Using Safra's determinization procedure~\cite{DBLP:conf/focs/Safra88}, we obtain a (deterministic) Rabin automaton $\mathcal C_\varphi^B$ that is language-equivalent to $\mathcal B_\varphi^B$. The automaton $\mathcal C_\varphi^B$ has $2^{5^{c_0|\cl(\varphi)|}}$ states and $5^{c_1 \cdot |\cl(\varphi)|}$ Rabin pairs where $c_0 > c_1$ are suitable constants.

\paragraph{\emph{Step 3}}
We construct the (unlabeled) product game graph $\mathcal G' = (V', E')$ of the game graph $\mathcal G = (V, E, \lambda)$ and the Rabin automaton $\mathcal C_\varphi^B = (Q, 2^\mathcal P, q_0, \delta, \Omega)$ such that $V' = V \times Q$ and
\[ \bigl( (v, q), (v', q') \bigr) \in E' \text{ if and only if $(v, v') \in E$ and $\delta(q, \lambda(v)) = q'$.} \]
Moreover, we define the Rabin winning condition of $\mathfrak G'$ to be
\[  \Omega' = \bigl\{ ((V, E), (V, F)) \in V' \times V' \mid (E, F) \in \Omega \bigr\}. \]
The desired Rabin game is then $\mathfrak G' = (\mathcal G', \Omega')$.

An induction over the length of a play $\rho' = (v_0, q_0) (v_1, q_1) \ldots$ in $\mathfrak G'$ shows that Player~0 wins $\rho'$ if and only if Player~0 wins the play $\rho = v_0 v_1 \ldots$ in $\mathfrak G$.

\paragraph{\emph{Step 4}}
Finally, by applying Piterman and Pnueli's method~\cite{DBLP:conf/lics/PitermanP06}, we solve the resulting Rabin game in time $\mathcal O(n^{k+3}kk!)$ where $n = |V| \cdot 2^{5^{c_0|\cl(\varphi)|}}$ is the number of vertices and $k = 5^{c_1|\cl(\varphi)|}$ is the number of Rabin pairs of $\mathfrak G'$.

In total, we obtain the following results.

\begin{theorem}
Given an rLTL$(\boxdot,\Diamonddot)$ game $\mathfrak G = (\mathcal G, (\varphi, B))$ with $\mathcal G = (V, E, \lambda)$ and a vertex $v_0 \in V$, one can 
\begin{enumerate}
	\item decide which player has a winning strategy from $v_0$ (i.e., Problem~\ref{prob:synthesis_decision}) and
	\item compute a winning strategy for the corresponding player (i.e., Problem~\ref{prob:synthesis_strategy})
	\end{enumerate}
in time $\mathcal O(n^{k+3}kk!)$ where $n = |V| \cdot 2^{5^{c_0|\cl(\varphi)|}}$, $k = 5^{c_1|\cl(\varphi)|}$, and $c_0, c_1$ are suitable constants.
\end{theorem}


\section{Full $\mathrm{r}$LTL}
\label{sec:full_rLTL}
In this section, we extend the semantics of rLTL$(\boxdot,\Diamonddot)$ to full rLTL by providing the semantics for three additional operators: next (denoted by $\Nextdot$), release (denoted by $\Releasedot$), and until (denoted by $\Untildot$). Moreover, we show that all the results obtained for rLTL$(\boxdot,\Diamonddot)$ easily extend to full rLTL. In particular, we present expansion rules for the dotted version of release and until, sketch how to construct equivalent Büchi automata from rLTL formulas, and revisit the model checking and synthesis problems in the setting of full rLTL.

\subsection{Robust Semantics of Next, Release, and Until}
\label{SSec:FullrLTL}

The robust semantics of next is a direct generalization of the LTL semantics from $\B$ to $\B_4$:
\[V(\sigma,\Nextdot\varphi)=V(\sigma_\suf{1},\varphi).\]

However, this is not the case for the release and until operators since they can be used to recover $\boxdot$ and $\Diamonddot$ via the equalities $\boxdot\psi \coloneqq \false \Releasedot \psi$ and $\Diamonddot \psi \coloneqq \true \Untildot \psi$, respectively, and $\boxdot$ and $\Diamonddot$ themselves are not a direct generalization of their LTL counterparts. 

In order to motivate the semantics of release, we return to our motivating example $\Box p$. According to the safety-progress classification of temporal properties, eloquently put forward in~\cite{SafetyProgress}, $\Box p$ defines a safety property. It can be expressed as $A(L)$ with $L$ being the regular language $(\true)^* p$ and $A$ the operator generating all the infinite words in $(2^\mathcal{P})^\omega$ with the property that all its finite prefixes belong to $L$. In addition to $A$, we can find in~\cite{SafetyProgress} the operators $E$, $R$, and $P$ defining guarantee, response, and persistence properties, respectively. The language $E(L)$ consists of all the infinite words that contain at least one prefix in $L$, the language $R(L)$ consists of all the infinite words that contain infinitely many prefixes in $L$, and the language $P(L)$ consists of all the infinite words such that all but finitely many prefixes belong to $L$. Using these operators we can reformulate the semantics of $\boxdot p$ as:
\begin{equation}
\label{SemanticsViaAPRE}
V(\sigma,\boxdot p) = \begin{cases}
	1111 & \text{if $\sigma \in A(L)$;} \\
	0111 & \text{if $\sigma \in P(L)\setminus A(L)$;} \\
	0011 & \text{if $\sigma \in R(L) \setminus \left( A(L)\cup P(L) \right)$;} \\
	0001 & \text{if $\sigma \in E(L) \setminus \left( A(L)\cup P(L) \cup R(L) \right)$; and} \\
	0000 & \text{if $\sigma \notin E(L)$.}
\end{cases}
\end{equation}

We thus obtain a different justification for the five different truth values used in rLTL and why the five different cases in~\eqref{Eq:CanonicalForms} can be seen as canonical. Equality~\eqref{SemanticsViaAPRE} also suggests how we can define the $5$-valued semantics for the release operator. Recall that the LTL formula $p\,\mathcal{R}\,q$, for atomic propositions $p$ and $q$, defines a safety property, and that its semantics is given by
\begin{equation}
\label{SemanticsRelease}
W(\sigma, p \Release q ) = \inf_{j\ge 0}\max\left\{V_1(\sigma_\suf{j},q),\sup_{0\le i<j}V_1(\sigma_\suf{i},p)\right\}.
\end{equation}
We can interpret
\[ \max\left\{V_1(\sigma_\suf{j},q),\sup_{0\le i<j}V_1(\sigma_\suf{i},p)\right\}\]
as the definition of the regular language $L=(\true)^* q+(\true)^*p(\true)^+$ and $\inf_{j\ge 0}$ as the requirement that every prefix of a string satisfying $p\,\mathcal{R}\,q$ belongs to $L$ (i.e., as the definition of the operator $A$). Therefore, the $5$-valued semantics can be obtained by successively enlarging the language $A(L)$ through the replacement of the operator $A$, formalized by $\inf$ in Equation~\eqref{SemanticsRelease}, by the operators $P$ formalized by $\sup \inf$, $R$ formalized by $\inf\sup$, and $E$ for formalized by $\sup$.  This observation leads to the semantics
\[
V(\sigma, \varphi \Releasedot \psi ) = \left(V_1(\sigma,\varphi \Releasedot \psi ),V_2(\sigma,\varphi \Releasedot \psi ),V_3(\sigma, \varphi \Releasedot \psi ),V_4(\sigma, \varphi \Releasedot \psi )\right),
 \]
where
\begin{align}
V_1(\sigma,\varphi \Releasedot \psi ) & = \inf_{j\ge 0}\max\left\{V_1(\sigma_\suf{j},\psi),\sup_{0\le i<j}V_1(\sigma_\suf{i},\varphi)\right\}, \\
V_2(\sigma,\varphi \Releasedot \psi ) & = \sup_{k\ge 0}\inf_{j\ge k}\max\left\{V_2(\sigma_\suf{j},\psi),\sup_{0\le i<j}V_2(\sigma_\suf{i},\varphi)\right\}, \\
V_3(\sigma, \varphi \Releasedot \psi ) & = \inf_{k\ge 0}\sup_{j\ge k}\max\left\{V_3(\sigma_\suf{j},\psi),\sup_{0\le i<j}V_3(\sigma_\suf{i},\varphi)\right\}, \\
V_4(\sigma,\varphi \Releasedot \psi ) & = \sup_{j\ge 0}\max\left\{V_4(\sigma_\suf{j},\psi),\sup_{0\le i<j}V_4(\sigma_\suf{i},\varphi)\right\}.
\end{align}
We note that $\boxdot\psi = \false \Releasedot \psi$ holds, thereby showing that the semantics for \Releasedot is compatible with the semantics of $\boxdot$ introduced in Section~\ref{Sec:ANew}. We can glean further intuition behind the definition of $\Releasedot$ by considering the special case $\varphi = p$ and $\psi = q$ for two atomic propositions  $p, q \in \mathcal P$. Expressing $V(\sigma, p \Releasedot q)$ in terms of an LTL valuation $W$, we obtain
\[V(\sigma, p \Releasedot q)=\left(W(\sigma, p \Release q),~ W(\sigma,\Diamond\Box q \lor \Diamond p),~ W(\sigma,\Box\Diamond q\lor \Diamond p),~ W(\sigma,\Diamond q\lor \Diamond p)\right).\]

We see that, as long as $p$ occurs, the value of $p \Releasedot q$ is at least $0111$. It could be argued that the semantics of $p \Releasedot q$ should also count the number of occurrences of $q$ preceding the first occurrence of $p$. As we detail in Section~\ref{SSection:ExamplesFullrLTL}, such property can be expressed in rLTL by making use of the proposed semantics.

In LTL, the until operator is dual to the release operator but such relationship does not extend to rLTL in virtue of how negation was defined. Hence, the semantics of $\Untildot$ has to be introduced independently of $\Releasedot$. We follow the same approach that was used for $\Releasedot$ by interpreting the LTL semantics of $p \Untildot q$, given by
\begin{equation}
\label{SemanticsUntil}
W(\sigma, p \Until q ) = \sup_{j\ge 0}\min\left\{V_1(\sigma_\suf{j},q),\inf_{0\le i<j}V_1(\sigma_\suf{i},p)\right\},
\end{equation}
as defining the language $E(p^*q)$. In the hierarchy of the operators $E$, $R$, $P$, and $A$, defined by the inclusions $A(L)\subset P(L)\subset R(L)\subset E(L)$ for any regular language $L$, the language $E(p^*q)$ cannot be enlarged as it sits at the top of the hierarchy. Therefore, the semantics of $\Untildot$ is given by
\[
V(\sigma, \varphi \Untildot \psi ) = \left(V_1(\sigma,\varphi \Untildot \psi ),V_2(\sigma,\varphi \Untildot \psi ),V_3(\sigma, \varphi \Untildot \psi ),V_4(\sigma, \varphi \Untildot \psi )\right),
 \]
where
\[
V_k(\sigma,\varphi \Untildot \psi ) = \sup_{j\ge 0}\min\left\{V_k(\sigma_\suf{j},\psi),\inf_{0\le i<j}V_k(\sigma_\suf{i},\varphi)\right\} \text{ for each } k\in\{1,2,3,4\}.
\]

We obtain, by definition, that the semantics of \Untildot is compatible with the semantics of $\Diamonddot$ in the sense that $\true \Untildot \psi= \Diamonddot \psi$.

\subsection{Examples}
\label{SSection:ExamplesFullrLTL}
As we discussed before, the semantics of $\varphi\Releasedot \psi$ does not count how many times $\psi$ holds before the first occurrence of $\varphi$. This property, however, is captured by the rLTL formula
\begin{equation}
\label{Eq:ReleasedotWithCounting}
\left(\varphi\Releasedot \psi\right)\land \left(\neg\varphi\Untildot \psi\right).
\end{equation} 
To see why, we assume $\varphi=p$ and $\psi=q$, for atomic propositions $p$ and $q$, so as to express the semantics of the rLTL formula~\eqref{Eq:ReleasedotWithCounting} in terms of an LTL valuation $W$ as
\begin{align}
V(\sigma,\left(p\Releasedot q\right)\land \left(\neg p\Untildot q\right)) & =\left(W(\sigma, p\Release q),W(\neg p\Until q),W(\neg p\Until q),W(\neg p\Until q)\right).
\end{align}
Note how we can now distinguish between three cases: $p\Release q$ holds, corresponding to value $1111$, $q$ holds at least once before being released by $p$, corresponding to value $0111$, and $q$ does not hold before being released by $p$, corresponding to value $0000$.

The preceding discussion showed how the LTL equality $\varphi \Release \psi=\left( \varphi \Release \psi\right)\land\left(\neg \varphi \Until \psi\right)$ is not valid in rLTL. Another LTL equality that is not valid in rLTL is the decomposition of the until operator into its liveness and safety parts given by
\[\varphi \Until \psi = \Diamond\psi\land (\psi \Release (\psi\lor\varphi)).\]
The rLTL formula $\Diamonddot\psi\land (\psi \Releasedot (\psi\lor\varphi))$ expresses a weaker requirement than $\varphi \Untildot \psi$ that is also useful to express robustness. When $\varphi$ and $\psi$ are the atomic propositions $p$ and $q$, respectively, the semantics of $\Diamonddot\psi\land (\psi \Releasedot (\psi\lor\varphi))$ can be expressed in terms of an LTL valuation W as
\[V(\sigma, \Diamond q\land (q \Release (q\lor p)) )=\left(W(\sigma, p \Until q),~ W(\sigma, \Diamond q),~ W(\sigma,\Diamond q),~ W(\sigma,\Diamond q)\right).\]
Whereas $\varphi \Untildot \psi$ only assumes two values, $\Diamonddot\psi\land (\psi \Releasedot (\psi\lor\varphi))$ assumes $3$ possible values allowing to separate the words that violate $\varphi \Until \psi$ into those that satisfy $\Diamond q$ and those that do not.

\subsection{From Full rLTL to Generalized Büchi Automata}
\label{sec:fullrLTL_2_BA}
The construction of a generalized Büchi automaton from an rLTL formula relies on the following expansion rules for $\Releasedot$ and $\Untildot$. Once can prove these rules using arguments similar to those employed to prove Proposition~\ref{prop:expansion_rules}.

\begin{proposition}[Expansion Rules for $\Releasedot$ and $\Untildot$]
For any rLTL formulas $\varphi$ and $\psi$, for any $\sigma\in \Sigma^\omega$, any $\ell\in\N$, and any valuation $V$ the following equalities hold:
\begin{align}
\label{Eq:RRExp1}
V_1(\sigma_\suf{\ell},\varphi \Releasedot \psi) &= \min\left\{V_1(\sigma_\suf{\ell},\psi),\max\left\{V_1(\sigma_\suf{\ell},\varphi),V_1(\sigma_\suf{\ell+1},\varphi \Releasedot \psi)\right\}\right\}\\
\label{Eq:RRExp2}
V_2(\sigma_\suf{\ell}, \varphi \Releasedot \psi) &= \max\left\{V_1(\sigma_\suf{\ell},\varphi \Releasedot \psi),V_2(\sigma_\suf{\ell},\varphi),V_2(\sigma_\suf{\ell+1},\varphi \Releasedot \psi)\right\}\\
\label{Eq:RRExp3}
V_3(\sigma_\suf{\ell}, \varphi \Releasedot \psi) &= \min\left\{V_4(\sigma_\suf{\ell},\varphi \Releasedot \psi),\max\left\{V_3(\sigma_\suf{\ell},\varphi),V_3(\sigma_\suf{\ell+1},\varphi \Releasedot \psi)\right\}\right\}\\
\label{Eq:RRExp4}
V_4(\sigma_\suf{\ell}, \varphi \Releasedot \psi) &= \max\left\{V_4(\sigma_\suf{\ell},\psi),V_4(\sigma_\suf{\ell},\varphi),V_4(\sigma_\suf{\ell+1},\varphi\Releasedot \psi)\right\}\\
\label{Eq:RRExp5}
V_1(\sigma_\suf{\ell}, \varphi \Untildot \psi) &= \max\left\{V_1(\sigma_\suf{\ell},\psi), \min\left\{V_1(\sigma_\suf{\ell},\varphi),V_1(\sigma_\suf{\ell+1},\varphi \Untildot \psi)\right\}\right\}\text{ for each }k\in\{2,3,4\}.
\end{align}
\end{proposition}

One can translate rLTL formulas into generalized Büchi automata by means of a straightforward extension of the rLTL$(\boxdot,\Diamonddot)$ construction introduced in Section~\ref{sec:LTL_2_BA}.
For this reason, we only sketch this extension:
\begin{itemize}
\item Logical connectives are handled as in rLTL$(\boxdot,\Diamonddot)$.
	\item Due to the simple semantics of the operator $\Nextdot$, this case is handled in the same manner that $\Next$ is handled in LTL (see, e.g., Baier and Katoen~\cite{Baier:2008:PMC:1373322}).
	\item The operator $\Releasedot$ is handled in the same manner as the operator $\boxdot$ (see Section~\ref{sec:LTL_2_BA}) while applying the expansion rules for $\Releasedot$ given by Equations~\eqref{Eq:RRExp1} to \eqref{Eq:RRExp4}.	
	\item The operator $\Untildot$ is handled in the same manner as the operator $\Diamonddot$ (see Section~\ref{sec:LTL_2_BA}) while applying the expansion rules for $\Untildot$ given by Equation~\eqref{Eq:RRExp5}.	
\end{itemize}
Note that the temporal operators $\boxdot$ and $\Diamonddot$ can either be recovered syntactically from $\Releasedot$ and $\Untildot$ in a preprocessing step or handled directly as described in Section~\ref{sec:LTL_2_BA}.
As in the case of rLTL($\boxdot, \Diamonddot$), we denote the Büchi automaton constructed from the formula $\varphi$ by $\mathcal A_\varphi$.

Although the expansion rules for \Releasedot, and \Untildot are different from the expansion rules for $\boxdot$ and $\Diamonddot$, a simple analysis yields that $\mathcal A_\varphi$ comprises $5^{|\cl(\varphi)|} + 4$ states and at most $4 \cdot |\cl(\varphi)|$ acceptance sets, exactly the same numbers as in the case of rLTL$(\boxdot,\Diamonddot)$. Moreover, $\mathcal A_\varphi$ exactly captures the semantics of $\varphi$ in the sense formalized below.

\begin{theorem}\label{thm:full_rLTL:buechi_automaton_correct}
Let $\varphi$ be an rLTL formula over the set $\mathcal P$ of atomic propositions, $\Sigma = 2^\mathcal P$, and $b \in \mathbb B_4$. Then, $\mathcal A_\varphi$ is a generalized Büchi automaton with $5^{|\cl(\varphi)|} + 4$ states and at most $4 \cdot |\cl(\varphi)|$ acceptance sets that accepts $\sigma \in \Sigma^\omega$ when starting in state $q_b$ if and only if $V(\sigma, \varphi) = b$.
\end{theorem}

\subsection{Model Checking and Synthesis}
Since we obtain the same bounds on the number of states and acceptance sets of the automaton $\mathcal A_\varphi$ for both rLTL($\boxdot, \Diamonddot$) formulas and full rLTL formulas, the results for model checking and synthesis extend to the case of full rLTL. For the reader's convenience, we provide the formal statements.

\begin{corollary}\label{thm:full_rLTL:model_checking}
One can decide the model checking problem as well as the at-least modecl checking problem for a generalized Büchi automaton $\mathcal A = (Q, \Sigma, q_0, \Delta, \mathcal F)$ and an rLTL formula $\varphi$ in time $\mathcal O \bigl( (|\mathcal F| + |\cl(\varphi)|) \cdot |Q| \cdot 5^{|\cl(\varphi)|} \bigr)$.
\end{corollary}

\begin{corollary}
Given an rLTL formula $\varphi$ over the set $\mathcal P$ of atomic propositions and a generalized Büchi automaton $\mathcal A = (Q, \Sigma, q_0, \Delta, \mathcal F)$ over the alphabet $2^\mathcal P$, one can compute the largest $b \in \mathbb B_4$ such that $V(\sigma, \varphi) \geq b$ for all $\sigma \in L(\mathcal A)$ in time $\mathcal O \bigl( (|\mathcal F| + |\cl(\varphi)|) \cdot |Q| \cdot 5^{|\cl(\varphi)|} \bigr)$.
\end{corollary}

\begin{corollary}
Given an rLTL game\footnote{An rLTL game is an rLTL$(\boxdot,\Diamonddot)$ game in which the winning condition is an rLTL formula.} $\mathfrak G = (\mathcal G, \varphi)$ with $\mathcal G = (V, E, \lambda)$ and a vertex $v_0 \in V$, one can 
\begin{enumerate}
	\item decide which player has a winning strategy from $v_0$ and
	\item compute a winning strategy for the corresponding player
	\end{enumerate}
in time $\mathcal O(n^{k+3}kk!)$ where $n = |V| \cdot 2^{5^{c_0|\cl(\varphi)|}}$, $k = 5^{c_1|\cl(\varphi)|}$, and $c_0, c_1$ are suitable constants.
\end{corollary}


\section{Quality is dual to robustness}
\label{Sec:Quality}
We motivated rLTL($\boxdot,\Diamonddot$) by the need to distinguish between the different ways in which safety properties can be violated. One can take a dual view and seek to distinguish between the different ways in which guarantee properties are satisfied. To illustrate this point, consider the LTL formula $\Diamond p\Rightarrow \Diamond q$ where $\Diamond p$ is an environment assumption and $\Diamond q$ is a system guarantee. According to the motto \emph{more is better} we would prefer the system to guarantee the stronger property $\Box\Diamond q$ whenever the environment satisfies the stronger property $\Box\Diamond p$. By now, the reader can already complete our argument: $\Diamond\Box p$ should lead to $\Diamond\Box q$ and $\Box p$ should lead to $\Box q$. Formalizing these ideas would still take us to a 5-valued logic where, however, negation needs to be defined differently. Although we can still use the linear order
\[ 0000\prec 0001\prec 0011\prec 0111\prec 1111 \]
on the set of truth values, one now needs to interpret the values differently. The value $0000$ still corresponds to \false but the remaining truth values now correspond to different quality values for \true with 0001 being the lowest quality and  1111 the highest. Negation, should then take $0000$ to $1111$ and all the remaining truth values to $0000$. Such negation is no more than the intuitionistic negation already discussed in Section~\ref{SSec:daCosta}, and would equip $\B_4$ with the structure of an Heyting algebra instead of the da Costa algebras used in this paper. This observation justifies the title of this section and suggests the following question: is there an extension of LTL that can be used to reason about \emph{both} robustness and quality? This is a question we will leave for further research.


\section{Discussion}
\label{sec:discussion}

The logic rLTL offers a transparent way to reason about the robustness of LTL specifications. Given an LTL formula $\varphi$, one obtains the corresponding rLTL formula $\psi$ simply by dotting the temporal operators in $\varphi$. The semantics of rLTL was constructed as a $4$-tuple whose first element corresponds to the LTL semantics of $\varphi$ and the remaining elements quantify by how much an infinite word violates $\varphi$. The technical development of the semantics was based on the insight that the temporal operators $\Box$ and $\Diamond$ count how often the formula they are applied to is satisfied thereby leading to a $5$-valued logic. We studied the verification and synthesis problems for rLTL and showed they can be solved in exponential and doubly exponential time, respectively. These complexity bounds are the same as those for LTL once we replace $2$, since LTL is Boolean valued, with $5$, since rLTL is $5$-valued. It remains an open problem to determine if these complexity upper bounds are tight. In addition to this question, we sketched in Section~\ref{Sec:Quality} a variant of rLTL tailored to quality and raised the question of how to combine robustness and quality in a single logic.

\bibliographystyle{alpha}
\bibliography{ComputerScience}

\end{document}